\documentclass[a4paper,twocolumn,11pt,unpublished]{quantumarticle}
\pdfoutput=1
\usepackage[english]{babel}
\usepackage[T1]{fontenc}
\usepackage[utf8]{inputenc}
\usepackage{lmodern,microtype}
\usepackage{amsmath,amssymb,amsthm,mathtools,bm}

\usepackage{graphicx,booktabs,array}
\usepackage{placeins}
\usepackage{enumitem}
\usepackage[numbers,sort&compress]{natbib}
\usepackage{xcolor}
\usepackage[colorlinks=true,allcolors=quantumviolet]{hyperref}
\usepackage[nameinlink,noabbrev]{cleveref}
\usepackage{array}
\newcolumntype{L}[1]{>{\raggedright\arraybackslash}p{#1}}
\makeatletter
\patchcmd{\@printtitle}{\@printtitletextwithappropriatefontsize}
  {\color{quantumviolet}\@printtitletextwithappropriatefontsize}{}
  {\PackageError{manuscript}{Unable to apply violet title styling}{}}
\makeatother
\newcolumntype{L}[1]{>{\raggedright\arraybackslash}p{#1}}

\hypersetup{pdftitle={Adaptive Record Consistency in a Schulman-Type Model},pdfauthor={D.M. Theshan N. Weerasinghe},pdfsubject={Composition, detector memory, and sharp calibration-disturbance bounds}}
\setlist{itemsep=2pt,topsep=4pt,leftmargin=*}
\numberwithin{equation}{section}
\newtheorem{theorem}{Theorem}[section]
\newtheorem{lemma}[theorem]{Lemma}
\newtheorem{proposition}[theorem]{Proposition}
\newtheorem{corollary}[theorem]{Corollary}
\theoremstyle{definition}
\theoremstyle{remark}
\DeclareMathOperator{\sech}{sech}
\DeclareMathOperator{\TV}{TV}
\DeclareMathOperator{\Var}{Var}
\DeclareMathOperator{\sgn}{sgn}
\DeclareMathOperator{\arcosh}{arcosh}
\newcommand{\E}{\mathbb E}
\newcommand{\T}{\mathbb T}
\newcommand{\dd}{\,\mathrm d}
\newcommand{\eps}{\varepsilon}
\newcommand{\ind}{\mathbf 1}
\newcommand{\target}{P^{0}}
\newcommand{\epstar}{\varepsilon_*}
\newcommand{\epszero}{\varepsilon_0}
\newcommand{\Csharp}{\mathcal C_{0.4}}
\newcommand{\norm}[1]{\left\lVert #1\right\rVert}
\title{Adaptive Record Consistency in a Schulman-Type Model}
\author{D.M. Theshan N. Weerasinghe}
\affiliation{Department of Electrical and Computer Systems Engineering, Faculty of Engineering, Monash University, Clayton}
\email[\newline]{theshan.dissanayakemudiyanselage1@monash.edu}
\begin{document}
\raggedbottom
\widowpenalty=10000
\clubpenalty=10000
\displaywidowpenalty=10000
\maketitle
\begin{abstract}
Adaptive measurements test whether a retrocausal model preserves observable records when experiments are connected by classical feedback. We establish exact and quantitative composition constraints for a finite-width, coplanar Schulman-type source with positive wrapped-Cauchy propagation weights. Raw segment composition admits an explicit adaptive signalling witness. A uniquely determined scalar correction preserves recorded continuations, and a detector with preparation memory implements this correction while retaining the source correlations. For common, independently prepared terminal detectors coupled to the source through arrival angles, exact calibration forces equal-mass antipodal atomic responses. Under setting-only raw feedback, the sharp minimum worst disturbance of an earlier record is $\epstar(0)=[2\cosh(2\Gamma)]^{-1}$, where $\Gamma$ is the total source-path width. We derive explicit finite-error bounds and prove linear small-error improvement. For four settings at $\Gamma=0.4$, a uniform reduction over arbitrary nonnegative response measures and a continuous polynomial certificate establish $\epstar(\delta)=\epstar(0)-\Csharp\delta+O(\delta^2)$, with $\Csharp\simeq9.351091771057010$. A positive five-atom detector attains this coefficient. The results quantify the compatibility of isolated-source accuracy and adaptive record preservation, and provide precise design conditions for retrocausal apparatus models with explicit preparation, memory, and control.
\end{abstract}
\section{Introduction}
\label{sec:intro}
A theory of measurement must predict more than the outcome table of a single apparatus arrangement. It must also specify what happens when a recorded result controls another measurement, when a detector is reused, and when a memory is reset while a separate record remains accessible. These operations are elementary in laboratory practice. Their consistency is therefore a substantive requirement for an account of quantum phenomena, including an account in which hidden variables depend on future measurement settings.

The relevant distinction is between future-setting dependence of an unobserved state and future-setting dependence of an already accessible record. The former can be part of a retrocausal explanation; the latter would alter observable statistics. A stable pointer does not settle this issue. In a probability law obtained by globally normalizing weights of entire histories, the pointer can retain the same bit within every history while a later apparatus changes the relative statistical weights of histories containing different bits. The resulting marginal record distribution can change even though no individual archive has been overwritten.

This paper develops a single quantitative study of that problem. We start with a specified positive-weight source, test its closure under adaptive composition, formulate the exact record-preservation condition, and examine an explicit memory-dependent continuation rule. We then enlarge the detector search to arbitrary nonnegative terminal response measures under a declared source interface. The resulting obstruction is exact, robust to small calibration errors, and asymptotically sharp at a benchmark width. This progression turns an isolated counterexample into a constraint on a complete class of apparatus responses.

\subsection{Relation to existing work}
The source construction is based on the anomalous-rotation model analysed by Almada \emph{et al.}~\cite{almada2016}. Their isolated-pair analysis motivates the composition problem studied here: identifying the operational extensions that preserve symmetry-based marginal cancellation. We specify the raw segment rule and feedback substitution explicitly, thereby fixing the apparatus class to which the composition results apply.

Wharton and Argaman~\cite{wharton2020}, Sec.~V.1.2, footnote~24, already describe an adaptive reconstruction that combines the preserved first-measurement marginal with a conditional model on each realized branch. Their prescription supplies the conditional reconstruction used as a positive composition benchmark below. We establish the distinct consequences of globally normalizing literally substituted raw weights, and quantify the detector constraints imposed by that composition rule. Argaman's spacetime-local framework~\cite{argaman2026} explicitly distinguishes observable signal causality from future-input dependence of internal variables and identifies the corresponding detector explanation as an open problem.

Operational closure under local processing of correlations is an established subject~\cite{allcock2009}. Our normalized-continuation proof uses its elementary probabilistic logic. Time symmetry likewise has a broader process-theoretic treatment~\cite{selby2022}; here we distinguish an endpoint reciprocity identity from reversal of an entire apparatus with memory. The present kernel-based certificate assumes neither the causal-discovery framework of Wood and Spekkens~\cite{wood2015} nor the computational-complexity assumptions of Brogioli~\cite{brogioli2024}. Entanglement-swapping questions, including the causal distinctions discussed by Price and Wharton~\cite{price2021}, motivate more extensive network tests but require joint instruments that are not constructed here.

The contribution is a unified set of composition witnesses, detector rigidity and stability theorems, and a matching continuous certificate for a sharp calibration--disturbance coefficient. Together, these results connect the operational composition problem to a fully specified optimization over positive detector measures. They establish both a universal lower bound within that class and a constructive detector family attaining its leading improvement.

\subsection{Principal results and their interpretation}
For a detector family, let $\delta$ bound the total-variation error of every calibrated isolated-source table, and let $\eps$ be the largest change of Alice's first-record probability between a feedback protocol and an actual fixed-setting experiment. The definitions and all apparatus restrictions are given in \cref{sec:class}. Writing $\epstar(\delta)$ for the infimum of $\eps$ over that class, the central quantitative conclusions are
\begin{align}
\epstar(0)&=\frac{1}{2\cosh(2\Gamma)},\label{eq:headline-exact}\\
\epstar(0)-\epstar(\delta)&=\Theta_\Gamma(\delta),\quad\delta\downarrow0,\label{eq:headline-order}\\
\epstar(\delta)&=\frac{1}{2\cosh(0.8)}-\Csharp\delta+O(\delta^2).\label{eq:headline-coeff}
\end{align}
The last line takes $\Gamma=0.4$, with the certified coefficient
\begin{equation}
\begin{gathered}
9.35109177105700958345410586084\\
{}<\Csharp<{}\\
9.35109177105700958345410586086.
\end{gathered}
\end{equation}
The coefficient theorem applies to the four-setting alphabet specified in \cref{sec:class}. Explicit finite-error lower bounds and constructive upper bounds complement this asymptotic expansion; those bounds provide the quantitative statements at nonzero calibration error.

The exact obstruction does not require Bell violation. It holds at every finite positive width, including widths at which the target source's CHSH value is below two. Its physical meaning is a conflict between two specifications within the declared apparatus interface: maintaining the isolated-source table and maintaining earlier-record statistics under the stipulated feedback coupling. Conversely, the conditional continuation construction shows that preserving earlier records does not require eliminating every hidden future-setting dependence.

\section{Finite-width source and operational conventions}
\label{sec:model}
\subsection{Angles, propagation, and boundary outcomes}
All angles belong to the circle $\T=\mathbb R/(2\pi\mathbb Z)$ and are coplanar Bloch-spin angles measured in radians. They are not optical polarization angles. Binary outcomes are $a,b,c\in\{+1,-1\}$, and an outcome $s$ at setting $q$ fixes the endpoint direction
\begin{equation}
\theta(q,s)=q+\frac\pi2(1-s).
\label{eq:theta}
\end{equation}
The wrapped-Cauchy propagation density of width $\nu>0$ is
\begin{equation}
\begin{aligned}
k_\nu(d)&=\frac{\nu}{\pi}\sum_{n\in\mathbb Z}\frac{1}{(d+2\pi n)^2+\nu^2}\\
&=\frac{\sinh\nu}{2\pi(\cosh\nu-\cos d)}.
\end{aligned}
\label{eq:kernel}
\end{equation}
It is positive, even, periodic, and normalized with respect to $\dd d$ on $\T$. Circular convolution adds widths:
\begin{equation}
\int_\T k_{\nu_1}(\phi-\alpha)k_{\nu_2}(\beta-\phi)\dd\phi
=k_{\nu_1+\nu_2}(\beta-\alpha).
\label{eq:convolution}
\end{equation}
The proof follows from the Fourier expansion in \cref{app:identities}. All results use strictly positive widths. Zero-width expressions below are limits of already normalized probabilities, not integrations against a presumed zero-width density.

\subsection{Source distribution and setting-dependent hidden posterior}
The source constrains its initial directions to $\lambda$ and $\lambda+\pi$. A uniform reference measure on $\lambda$ is used in the unnormalized history construction. With fixed arm widths $\gamma_A,\gamma_B>0$, set
\begin{equation}
\Gamma=\gamma_A+\gamma_B,\qquad v=\sech\Gamma.
\end{equation}
Absorbing a common reference-measure constant into the weight, integration over $\lambda$ gives
\begin{equation}
\begin{aligned}
&J_\Gamma(a,b\mid x,y)\\
&\quad=\int_\T k_{\gamma_A}(\theta(x,a)-\lambda)\\
&\qquad\quad\times k_{\gamma_B}(\theta(y,b)-\lambda-\pi)\dd\lambda\\
&\quad=k_\Gamma\left(x-y+\frac\pi2(1+ab)\right).
\end{aligned}
\label{eq:Jsource}
\end{equation}
Normalize the four outcome weights to obtain
\begin{equation}
\target_{xy}(a,b)=\frac14[1-ab\sech\Gamma\cos(x-y)].
\label{eq:target}
\end{equation}
Both local marginals are $1/2$. The maximal planar CHSH value is $2\sqrt2\sech\Gamma$, exceeding two when $\Gamma<\arcosh\sqrt2$.

Define the antipodal sum
\begin{equation}
R_\nu(d)=k_\nu(d)+k_\nu(d+\pi)
=\frac{\sinh\nu\cosh\nu}{\pi(\cosh^2\nu-\cos^2d)}.
\label{eq:R}
\end{equation}
It is $\pi$-periodic but not constant for finite $\nu$. The source partition weight is $\sum_{a,b}J_\Gamma=2R_\Gamma(x-y)$. In particular, constant normalized marginals do not imply a setting-independent partition weight.

The uniform reference measure is not an assumption that the normalized hidden source state is independent of settings. Before conditioning on outcomes, that state satisfies
\begin{equation}
p(\lambda\mid x,y)\ \propto\ 
R_{\gamma_A}(x-\lambda)R_{\gamma_B}(y-\lambda-\pi).
\label{eq:hidden-source}
\end{equation}
It generally depends on both future endpoints. Confusing the reference measure with this posterior would remove the feature that the model is intended to retain.

\subsection{Operational definitions of composition and accessible records}
An accessible record is an outcome or archive entry that can be read by a permitted observer or controller. A hidden angle is not automatically an accessible record. A later policy $u$ selects apparatus settings using only the earlier records available to that controller. Every allowed final outcome is included in probability calculations. When a policy is written as $z=f_u(b)$, the selected setting is a derived variable constrained by the recorded $b$, rather than an additional freely specified input.

We distinguish two consistency tests. A spacelike no-signalling test compares a remote accessible marginal under different local controls outside its past light cone. A temporal record test compares an earlier marginal under changes to a causally later apparatus, possibly after an ordinary forward message. The explicit continuation witness below can be arranged as the first test. The terminal-detector impossibility theorem uses the second. Neither test is a claim that a pointer physically rewrites its own archived value.

\section{Adaptive composition and the record-preservation criterion}
\label{sec:composition}
\subsection{Record marginals under global normalization}
Let $h$ denote a finite accessible record and $\lambda$ the remaining hidden variables at a cut. Let $w(h,\lambda)\ge0$ be a reference weight with finite nonzero record weights
\begin{equation}
\begin{aligned}
W(h)&=\int w(h,\lambda)\dd\lambda>0,\\
p_0(h)&=\frac{W(h)}{\sum_{h'}W(h')},\\
\mu_h(\dd\lambda)&=\frac{w(h,\lambda)\dd\lambda}{W(h)}.
\end{aligned}
\end{equation}
After summing all future outcomes and hidden apparatus variables, let $F_u(\lambda,h)\ge0$ be the integrated future weight. Then
\begin{equation}
p_u(h)=\frac{p_0(h)A_u(h)}{\sum_{h'}p_0(h')A_u(h')},
\qquad A_u(h)=\E_{\mu_h}F_u.
\label{eq:cut-marginal}
\end{equation}
Only histories with positive reference weight are included, and the relevant expectations are assumed finite and positive.

\begin{theorem}[Exact record condition]\label{thm:cut}
The original record law is preserved under policy $u$ if and only if one common number $\kappa(u)>0$ satisfies
\begin{equation}
\int w(h,\lambda)F_u(\lambda,h)\dd\lambda
=\kappa(u)W(h)\quad\forall h.
\label{eq:record-condition}
\end{equation}
More generally, independence of the complete record law from $u$, without requiring agreement with $p_0$, is equivalent to $A_u(h)=\kappa(u)r(h)$ for one positive function $r$.
\end{theorem}
\begin{proof}
The ratio of probabilities for records $h,h'$ equals
$[p_0(h)/p_0(h')][A_u(h)/A_u(h')]$. Independence of every ratio is precisely proportionality of the positive vectors $A_u$. Agreement with $p_0$ requires their entries to be equal. The converses follow by cancelling the common factor in \cref{eq:cut-marginal}.
\end{proof}

After rescaling $\kappa$ to one, it is permissible to have
\begin{equation}
F_u=1+g_u,\qquad \E_{\mu_h}g_u=0,\qquad 1+g_u\ge0.
\label{eq:zero-mean}
\end{equation}
The record law is unchanged while the conditional hidden measure becomes $(1+g_u)\mu_h$. Pointwise constancy in $\lambda$ is sufficient but stronger than necessary. If $h$ includes hidden variables, demanding preservation of its full distribution imposes a correspondingly stronger condition than operational record preservation.

\subsection{Raw continuation and an exact adaptive signalling witness}
Bob first measures at $y$, records $b$, and continues the particle from endpoint $\theta(y,b)$. At a second axis $z=f_u(b)$, an added segment of width $\eta>0$ has weight
\begin{equation}
K_\eta(c,b;y,z)=k_\eta\bigl(\theta(z,c)-\theta(y,b)\bigr).
\end{equation}
The \emph{raw extension} assigns
\begin{equation}
\begin{aligned}
&p_{\mathrm{raw}}(a,b,c\mid x,y,u)\\
&\quad=\frac{J_\Gamma(a,b\mid x,y)K_\eta(c,b;y,f_u(b))}
{\mathcal Z_u(x,y)},\\
&\mathcal Z_u(x,y)\\
&\quad=\sum_{a',b',c'}J_\Gamma(a',b'\mid x,y)\\
&\qquad\qquad\times K_\eta(c',b';y,f_u(b')).
\end{aligned}
\label{eq:raw}
\end{equation}
The controller imposes its branch relation without an extra branch weight. This completes the rule being tested.

Summing $c$ gives $R_\eta(z-y)$, independent of $b$ if $z$ is fixed. Backward summation therefore preserves the source record distribution under any finite raw continuation with all axes and widths fixed independently of outcomes. Under adaptive settings, the row mass can differ between the two $b$ branches.

Set $x=y=0$. Compare $f_0(+)=f_0(-)=\pi/2$ with $f_1(+)=\pi/2$, $f_1(-)=0$. Writing $M_b(u)=R_\eta(f_u(b))$, direct summation gives
\begin{equation}
p_{\mathrm{raw}}(a=+1\mid u)
=\frac12+\frac{\sech\Gamma}{2}\frac{M_-(u)-M_+(u)}{M_-(u)+M_+(u)}.
\end{equation}
The ratio $R_\eta(0)/R_\eta(\pi/2)=\coth^2\eta$ yields
\begin{equation}
\begin{aligned}
\Delta_A&:=p_{\mathrm{raw}}(+\mid1)-p_{\mathrm{raw}}(+\mid0)\\
&=\frac{1}{2\cosh\Gamma\cosh(2\eta)}.
\end{aligned}
\label{eq:raw-witness}
\end{equation}
For equal original arms, $\Gamma=2\gamma$, recovering the equal-width expression. Alice's measurement can be spacelike separated from Bob's protocol choice and both of his measurements. No Bob outcome is selected or discarded; Alice's own complete outcome sequence has a different predicted marginal. The elementary evenness of $k_\eta$ remains true, but flipping $b$ changes the selected axis, so it no longer pairs equally weighted complete histories of this controlled experiment.

\subsection{Uniqueness of scalar continuation normalization}
Consider corrected weights
\begin{equation}
D_z(c\mid b;y)=m(b,y,z)K_\eta(c,b;y,z),\qquad m>0,
\end{equation}
where $m$ does not depend on $c$ or on remote inputs. Fix $y,\eta$ and allow independent choices of the two branch settings $z_+,z_-$.
\begin{proposition}[Scalar uniqueness]\label{prop:scalar}
For a finite-width source with $x=y$, preservation of Alice's marginal $1/2$ for every such policy is equivalent to
\begin{equation}
m(b,y,z)R_\eta(z-y)=\kappa
\label{eq:scalar-condition}
\end{equation}
for every permitted $(b,z)$, with one common positive $\kappa$.
\end{proposition}
\begin{proof}
Since $\sech\Gamma>0$, the marginal formula forces equality of the two corrected row masses for every pair $(z_+,z_-)$. Fix one branch and vary the other, then reverse their roles. Both row-mass functions must equal the same constant. Conversely, that constant cancels from the normalized joint law.
\end{proof}
Choose $\kappa=1$. The resulting transition is
\begin{equation}
\begin{aligned}
T_\eta(c\mid b;y,z)&=\frac{K_\eta(c,b;y,z)}{R_\eta(z-y)}\\
&=\frac12[1+bc\sech\eta\cos(z-y)].
\end{aligned}
\label{eq:transition}
\end{equation}
Summing $c$ preserves the entire original $(a,b)$ distribution for arbitrary $x,y$, beyond the equal-axis witness used for necessity. The transition is reciprocal under exchanging the endpoint data:
\begin{equation}
T_\eta(c\mid b;y,z)=T_\eta(b\mid c;z,y).
\label{eq:reciprocity}
\end{equation}
For fixed axes its binary matrix is doubly stochastic and obeys detailed balance for the uniform prior. This identity does not specify reversal of a controller, latch, archive, reservoir, or nonuniform boundary state.

If corrected masses are instead $M_b=\kappa(1+e_b)$ with $|e_b|\le\zeta<1$, the maximum absolute imbalance ratio is $\zeta$. Hence the equal-axis marginal obeys
\begin{equation}
\left|p(a=+1)-\frac12\right|\le\frac12\sech\Gamma\,\zeta.
\end{equation}
Opposite extreme errors attain this bound. The discrepancy between two protocols satisfying the same error envelope is at most $\sech\Gamma\,\zeta$.

\subsection{Finite adaptive closure and full local transcripts}
\begin{theorem}[Normalized continuation closure]\label{thm:closure}
Start from a normalized law of initial records that is independent of later controls. Append a finite acyclic sequence of conditional operations $T_j(c_j\mid h_j,u)$, each normalized for every allowed earlier history and each using only the controls available at that operation. Suppose changing a later control does not change an earlier factor. Then the product construction is normalized and preserves every earlier complete record prefix. Finite elimination is associative if the retained interface includes every record needed by subsequent controllers.
\end{theorem}
\begin{proof}
Sum the last output in $p_0\prod_j T_j$; its factor becomes one for every earlier history. Repeat in reverse topological order. Stopping at any prefix proves its preservation. Associativity follows by regrouping finite sums. If a later operation uses an eliminated bit, its controller must be included in the eliminated block or the bit must remain in the interface; otherwise the experiment has changed.
\end{proof}

For normalized local transcript kernels $Q_A(t_A\mid a,x,u_A)$ and $Q_B(t_B\mid b,y,u_B)$,
\begin{equation}
\begin{aligned}
&p(t_A,t_B\mid x,y,u_A,u_B)\\
&\quad=\sum_{a,b}\target_{xy}(a,b)Q_A(t_A\mid a,x,u_A)\\
&\qquad\qquad\times Q_B(t_B\mid b,y,u_B).
\end{aligned}
\end{equation}
Summing Bob's transcript leaves $\frac12\sum_aQ_A(t_A\mid a,x,u_A)$, independent of $y,u_B$. The same holds with the parties exchanged. The argument also permits independent source copies and normalized local processing of their recorded outputs, with the original source settings fixed as external inputs. It supplies neither a hidden joint-measurement kernel nor closure under arbitrary entangling devices or causal loops. Marginalizing an actual intermediate measurement retains its disturbance; it is not equivalent to removing that apparatus physically.

\section{Detector memory and composition requirements}
\label{sec:memory}
\subsection{Invariance of branch weights under passive memory}
\begin{proposition}[Passive-memory obstruction]\label{prop:passive}
Decorate a fixed raw segment $K_u(c,\lambda,h)$ with memory variables $m$ through a conditional density $B_u(m\mid c,\lambda,h)$ satisfying $\int B_u\dd m=1$. If no other weights or accessible outcomes change, this decoration cannot alter the record-preservation condition.
\end{proposition}
\begin{proof}
The integrated future contribution is unchanged:
\begin{equation}
\begin{aligned}
&\sum_c\int K_u(c,\lambda,h)B_u(m\mid c,\lambda,h)\dd m\\
&\qquad=\sum_cK_u(c,\lambda,h).
\end{aligned}
\end{equation}
Insert this identity into \cref{eq:record-condition}.
\end{proof}
Deterministic pointer copying, redundant archives, and normalized noise channels on those copies are included. An active change to the interaction, outgoing state, or boundary weighting is not included. The proposition concerns marginalization of the stipulated factors, not a universal limitation on memory or decoherence.

\subsection{An effective detector with preparation memory}
Let a carried tag be $r=\bot$ for an unmeasured source arm or $r=(q,b)$ for a recorded preparation. In the latter case the incoming segment starts at $\alpha=\theta(q,b)$. An archive $H$ retains previous records. At setting $z$, define the effective segment-integrated detector weight
\begin{align}
&D_\nu(c,r',H'\mid\alpha,r,H;z)\nonumber\\
&\quad=k_\nu(\theta(z,c)-\alpha)\,C_\nu(r,z)
\ind_{r'=(z,c)}\ind_{H'=H\Vert(z,c)},\label{eq:tag-detector}\\
&C_\nu(r,z)=
\begin{cases}1,&r=\bot,\\ R_\nu(z-q)^{-1},&r=(q,b).
\end{cases}\label{eq:tag-factor}
\end{align}
Incompatible recorded tags and starting endpoints have zero weight. The effective model uses exact angle memory and a common detector formula with distinct blank and recorded input states. Its specified branch factors give the target for a microscopic implementation with preparation-dependent interactions.

\begin{proposition}[Protected-tag success]\label{prop:tag-success}
With externally fixed original source settings and blank source tags, finite adaptive local continuations using \cref{eq:tag-detector}, with valid tags retained, preserve all earlier record prefixes and the original hidden source density after the appended outcomes are marginalized.
\end{proposition}
\begin{proof}
A recorded input reduces to \cref{eq:transition}. Its deterministic archive and tag updates add no multiplicity. Reverse summation over the added trees yields a factor one for every fixed $(a,b,\lambda)$ of the original source history. Prefix preservation and full-transcript no-signalling then follow as in \cref{thm:closure}.
\end{proof}

This construction also permits a nonconstant hidden future factor at an interior cut. Split a continuation width as $\nu=\nu_1+\nu_2$. The conditional bridge is
\begin{equation}
q(\phi\mid b,c;q,z)=
\frac{k_{\nu_1}(\phi-\theta(q,b))k_{\nu_2}(\theta(z,c)-\phi)}
{k_\nu(\theta(z,c)-\theta(q,b))}.
\end{equation}
After combining with $T_\nu$ and summing $c$,
\begin{equation}
q(\phi\mid b;q,z)=
k_{\nu_1}(\phi-\theta(q,b))\frac{R_{\nu_2}(z-\phi)}{R_\nu(z-q)}.
\label{eq:bridge}
\end{equation}
Convolution normalizes this expression. It generally depends on the future setting $z$ while the earlier accessible record law does not. Measuring $\phi$ would require another specified detector rule.

\subsection{Display resets and preparation-memory erasure}
Reset the carried tag to $\bot$ without changing the spin or an independent archive, and assign unit total weight to the reset. The next detector now uses its raw branch. With $x=y=0$, $z(+)=\pi/2$, $z(-)=0$, the earlier marginal becomes
\begin{equation}
p(a=+1)=\frac12+\frac{1}{2\cosh\Gamma\cosh(2\nu)}.
\end{equation}
At $\Gamma=0.4$, $\nu=0.2$, this is approximately $0.92781939$, compared with $0.5$ for an intact tag. Every history still contains its old archived bit. The failure is statistical reweighting.

A display-only reset is different. If the old display is logged in an environment and the preparation data $(q,b)$ remain available, the normalized continuation law is unchanged. Passing that reset does not prove closure under erasure of all preparation data. Conversely, the tag-only reset counterexample does not establish that a physical reset can leave every other variable and weight unchanged.

\subsection{Adaptive initial measurement of the partner particle}
Suppose Alice measures first at $x=0$ and sends her outcome causally to Bob. Bob then measures his still-unmeasured partner using $y(+)=\pi/2$ and $y(-)=0$. Both source tags are blank. The two realized Alice branch masses are $R_\Gamma(y(a))$, giving
\begin{equation}
p(a=+1)=\frac{R_\Gamma(\pi/2)}{R_\Gamma(\pi/2)+R_\Gamma(0)}
=\frac{1-\sech(2\Gamma)}{2}.
\label{eq:partner-failure}
\end{equation}
At $\Gamma=0.4$, the probability is $0.12615004$. This is a temporal consistency failure in an experiment containing an ordinary forward message, not the spacelike witness of \cref{eq:raw-witness}. A tag that remembers only earlier detections of the same particle is insufficient.

Within a scalar class preserving Bob's conditional outcome ratios, the independent-branch argument forces
\begin{equation}
m(a,x,y)=\frac{\kappa}{R_\Gamma(y-x)}.
\end{equation}
At unit scale, conditional reconstruction gives
\begin{equation}
\begin{aligned}
p(a,b\mid x,f)&=\frac12\frac{J_\Gamma(a,b\mid x,f(a))}{R_\Gamma(f(a)-x)}\\
&=\frac14[1-ab\sech\Gamma\cos(x-f(a))].
\end{aligned}
\label{eq:partner-repair}
\end{equation}
The required effective preparation information is $(x,-a,\Gamma)$, whereas same-particle continuation uses $(y,b,\eta)$. The total width $\Gamma$, not just Bob's leg width, is required. Multiplying by this scalar leaves the hidden bridge conditional on each realized endpoint pair unchanged.

No positive scalar $m(a,y)$ lacking $x$ can do this for all $x$ and all independently selectable branch settings. Equal branch settings first force $m(+,y)=m(-,y)=m(y)$. Preservation then requires $m(y)R_\Gamma(x-y)=C(x)$. Integrating over $x$ forces $m(y)$ constant because $\int_\T R_\Gamma(x-y)\dd x=2$, which would make $R_\Gamma$ constant, a contradiction. This argument specifies missing information in a restricted correction class; it does not prove a microscopic acquisition mechanism.

\subsection{Competing-clock implementation and source normalization}
For a known recorded input, two independent exponential clocks with rates
\begin{equation}
r_c=r_* k_\nu(\theta(z,c)-\theta(q,b)),\qquad r_*>0,
\end{equation}
have first-click density $p(c,t)=r_c\exp[-(r_++r_-)t]$. Integrating over $t\ge0$ yields $r_c/(r_++r_-)=T_\nu(c\mid b;q,z)$. A latch can store the winning channel. This is a stochastic implementation of the normalized conditional sector, assuming the rates. At finite observation time $T$, the no-click outcome with probability $\exp[-(r_++r_-)T]$ must also be included.

Using the same forward normalization on both original source legs at every fixed $\lambda$, while retaining a setting-independent uniform source distribution, gives
\begin{equation}
\begin{aligned}
t_A(a\mid\lambda,x)&=\tfrac12[1+a\sech\gamma_A\cos(x-\lambda)],\\
t_B(b\mid\lambda,y)&=\tfrac12[1-b\sech\gamma_B\cos(y-\lambda)],\\
p_{\mathrm{all}}(a,b\mid x,y)&=\tfrac14[1-ab\,v_{\mathrm{all}}\cos(x-y)],\\
v_{\mathrm{all}}&=\tfrac12\sech\gamma_A\sech\gamma_B.
\end{aligned}
\label{eq:all-normalized}
\end{equation}
This is an explicit Bell-local mixture. For equal arms its maximal planar CHSH value is $\sqrt2\sech^2\gamma<2$. At $\gamma=0.2$ it is approximately $1.359120$, versus $2.616316$ for the original source. Normalizing every hidden transition therefore does not solve the common-detector problem with that unchanged source measure.

\subsection{Hidden-state dependence of repeated-use consistency}
Let $p_0(h,\lambda)=1/4$ for $h\in\{0,1\}$ and $\lambda\in\{+1,-1\}$. For $0<\zeta<1$ and $s\in\{+1,-1\}$, consider positive nondisturbing factors
\begin{equation}
F_s(\lambda)=1+\zeta s\lambda,\qquad G(\lambda,h)=1+\zeta(-1)^h\lambda.
\end{equation}
Each has conditional mean one separately, but their product has mean $1+\zeta^2s(-1)^h$. Thus
\begin{equation}
p(h=0\mid s)=\frac{1+\zeta^2s}{2}.
\end{equation}
At $\zeta=0.4$, this gives $0.58$ and $0.42$. The first factor changes the hidden distribution on which the second was calibrated. The covariance term cannot be inferred from either isolated record test.

More generally, if $F>0$, $\E_\mu F=1$, and $\E_\mu F^2<\infty$, a controller able to apply $F$ twice to a frozen hidden state on one record branch and the identity on another can preserve both branch weights only if $F=1$ almost surely. The identity branch fixes the multiplier to one, so $\E F^2=1$ and $\Var(F)=0$. Hidden-state updates or restricted reuse can evade this lemma, but must be modelled. The one-cut condition is not itself a full composition law.

\section{An exact obstruction for terminal detectors}
\label{sec:class}
\subsection{Terminal reduction and structural assumptions}
The preceding examples leave open whether a different active detector could retain the source table and preserve adaptive records. We now allow all terminal responses compatible with a fixed interface. Integrating the antipodal source gives the arrival-angle coupling
\begin{equation}
\begin{aligned}
S_\Gamma(\phi,\psi)&=k_\Gamma(\phi-\psi+\pi)\\
&=\frac{C_\Gamma}{c+\cos(\phi-\psi)},\\
c&=\cosh\Gamma,\qquad C_\Gamma=\frac{\sinh\Gamma}{2\pi}.
\end{aligned}
\label{eq:S}
\end{equation}
For a setting $x$ and output $a$, let $d_{x,a}$ be a finite nonnegative measure on the arrival angle. It can arise by integrating an arbitrary positive local interaction over outgoing spin, memory, environment, and other terminal variables. Atomic measures are allowed to include the ideal boundary response.

The theorem class has the following structural restrictions:
\begin{enumerate}[label=(D\arabic*)]
\item \label{D1} The two arms use the same first-use terminal response family $d_{x,a}$. Their input apparatus states are independently and identically prepared for the same setting.
\item \label{D2} The arrival angle is the only source interface to each detector. No additional source-correlated preparation variable or shared apparatus environment bypasses it.
\item \label{D3} Terminal boundaries are fixed and unselected. Every apparatus output included in the model is summed; later constraints do not condition on a chosen terminal subset.
\item \label{D4} The source coupling and widths are fixed across the compared protocols. The feedback $y=f(a)$ changes the later detector only by substituting that numerical setting into the same response family, with no extra controller branch factor.
\end{enumerate}
These restrictions include active local apparatus dynamics after terminal reduction. They do not cover a common microscopic interaction receiving different source-correlated inputs in different experiments.

Define
\begin{align}
W^{ab}_{xy}&=\iint_{\T^2}S_\Gamma(\phi,\psi)d_{x,a}(\dd\phi)d_{y,b}(\dd\psi),\\
Z_{xy}&=\sum_{a,b}W^{ab}_{xy},\qquad P^{ab}_{xy}=W^{ab}_{xy}/Z_{xy}.
\label{eq:terminal-weights}
\end{align}
The identical response family and symmetric kernel imply $Z_{xy}=Z_{yx}$. Exact calibration demands $P_{xy}=\target_{xy}$ for every chosen setting pair, including equal-setting experiments. CHSH cross-pairs alone would not impose the diagonal extremality used below.

For all quantitative optimization results take
\begin{equation}
\mathcal X=\{0,\pi/4,\pi/2,3\pi/4\}.
\label{eq:alphabet}
\end{equation}
Replacing $3\pi/4$ with $-\pi/4$ and interchanging that setting's output labels gives the earlier equivalent alphabet. Calibration and the complete set of feedback policies are invariant under this relabeling.

\subsection{Extremal equal-setting correlation and rigidity}
The strict kernel extrema are
\begin{equation}
\ell=\frac{C_\Gamma}{c+1}=k_\Gamma(\pi),\qquad
u=\frac{C_\Gamma}{c-1}=k_\Gamma(0).
\label{eq:extrema}
\end{equation}
Here $u$ is the upper kernel value, not a protocol label. The minimum of $S_\Gamma$ occurs precisely at equal arrival angles and the maximum at antipodal ones.

\begin{theorem}[Equal-setting rigidity]\label{thm:rigidity}
For nonzero finite nonnegative terminal responses at setting $x$,
\begin{equation}
E_{xx}:=\sum_{a,b}abP^{ab}_{xx}\ge-\sech\Gamma.
\end{equation}
Equality holds if and only if
\begin{equation}
d_{x,+}=s_x\delta_{\vartheta_x},\qquad
d_{x,-}=s_x\delta_{\vartheta_x+\pi},\qquad s_x>0.
\label{eq:atomic}
\end{equation}
\end{theorem}
\begin{proof}
Let $A=\int d_{x,+}$ and $B=\int d_{x,-}$. When either mass vanishes the correlation is $+1$, so equality requires both positive. The kernel bounds give
\begin{equation}
W^{++}_{xx}+W^{--}_{xx}\ge\ell(A^2+B^2),\qquad
2W^{+-}_{xx}\le2uAB.
\end{equation}
Symmetry implies $W^{+-}_{xx}=W^{-+}_{xx}$. Therefore
\begin{equation}
r:=\frac{2W^{+-}_{xx}}{W^{++}_{xx}+W^{--}_{xx}}
\le\frac{2uAB}{\ell(A^2+B^2)}\le\frac u\ell.
\end{equation}
Since $(1-r)/(1+r)$ decreases with $r$, $E_{xx}\ge(\ell-u)/(\ell+u)=-1/c$.

Equality requires $A=B$, both self-integrals to attain their lower bounds, and the cross-integral its upper bound. Since $S_\Gamma-\ell$ is nonnegative and strictly positive away from the diagonal, each self-product measure must be supported on the diagonal. A nonzero finite measure with this property is a point mass: two disjoint positive-mass sets would give positive off-diagonal product mass. Cross saturation places the two atoms antipodally. Their masses agree. Conversely, these atoms attain all inequalities.
\end{proof}

Exact source calibration saturates this theorem at each setting. Substitution yields
\begin{align}
W^{ab}_{xy}&=s_xs_y k_\Gamma\left(\vartheta_x-\vartheta_y+\frac\pi2(1+ab)\right),\\
Z_{xy}&=2s_xs_yR_\Gamma(\vartheta_x-\vartheta_y),\label{eq:calibrated-Z}\\
P^{ab}_{xy}&=\tfrac14[1-ab\sech\Gamma\cos(\vartheta_x-\vartheta_y)].
\end{align}
All calibrated cosines must therefore agree: $\cos(\vartheta_x-\vartheta_y)=\cos(x-y)$. On the stated alphabet these orientations are related to the canonical ones by one common rotation or reflection. Arbitrary positive scales $s_x$ remain invisible in each fixed-setting normalized table.

This also derives antipodal evenness of the total convolved response. It was not imposed on the unknown detector. A regular response of nonzero angular width cannot attain the exact boundary value, but such responses are included in the approximate-calibration analysis.

\subsection{The sharp minimax first-record bias}
Let Alice record $a$ first and let Bob use the causal policy $f(+)=0$, $f(-)=\pi/2$. For exactly calibrated detectors the branch with outcome $a$ has total weight $Z_{x,f(a)}/2$. Put
\begin{equation}
t=\frac{s_{\pi/2}}{s_0},\qquad
\rho=\frac{R_\Gamma(0)}{R_\Gamma(\pi/2)}=\coth^2\Gamma>1.
\end{equation}
For Alice's settings $x=0$ and $x=\pi/2$, respectively,
\begin{equation}
p_0^f(+)=\frac\rho{\rho+t},\qquad
p_{\pi/2}^f(+)=\frac1{1+\rho t}.
\label{eq:exact-feedback}
\end{equation}
\begin{theorem}[Sharp exact obstruction]\label{thm:exact}
Within (D1)--(D4), exact calibration on an alphabet containing two orthogonal settings and their diagonal experiments forces
\begin{equation}
\begin{aligned}
&\inf_{t>0}\max\left\{\left|\frac\rho{\rho+t}-\frac12\right|,
\left|\frac1{1+\rho t}-\frac12\right|\right\}\\
&\qquad=\frac{\rho-1}{2(\rho+1)}=\frac1{2\cosh(2\Gamma)}.
\end{aligned}
\label{eq:exact-minimax}
\end{equation}
For the full four-setting disturbance objective below, equal-scale canonical atomic responses attain the same value. Thus $\epstar(0)$ equals \cref{eq:exact-minimax}.
\end{theorem}
\begin{proof}
If $t\le1$, the first probability in \cref{eq:exact-feedback} is at least $\rho/(\rho+1)$. If $t\ge1$, the second is at most $1/(\rho+1)$. Either case supplies the stated bias. At $t=1$ both attain it. With equal scales, every fixed marginal is $1/2$ and all feedback branch ratios are ratios of values of $R_\Gamma$. Its maximum and minimum on the alphabet occur at diagonal and orthogonal separations, so no other policy has larger bias. Rigidity has already exhausted all exactly calibrated response shapes.
\end{proof}

\Cref{fig:exact} shows the width dependence of the sharp obstruction.
\begin{figure}[tbp]
\centering\includegraphics[width=\columnwidth]{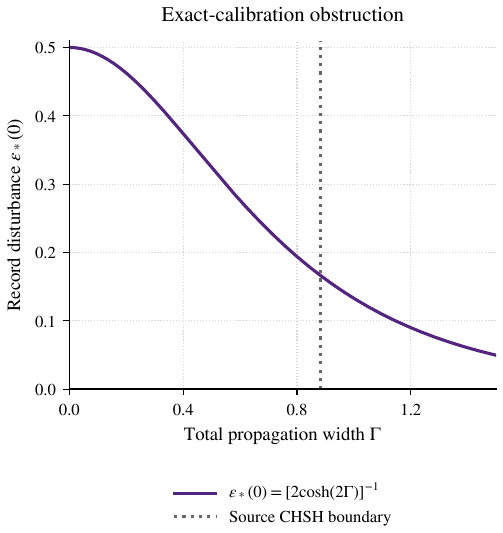}
\caption{The exact minimum worst first-record disturbance $\epstar(0)$ versus total propagation width $\Gamma$. The dotted line marks the source CHSH boundary $\Gamma=\arcosh\sqrt2$. The detector obstruction remains positive on both sides for every finite width.}
\label{fig:exact}
\end{figure}

At equal arm width $\gamma=0.2$, $\Gamma=0.4$, the target visibility is $0.92500745$ and the CHSH value is $2.61631617$. Yet the best worst bias is $0.3738499591$: the two earlier probabilities at equal scales are $0.8738499591$ and $0.1261500409$. Adjusting a setting-only scale transfers the failure from one first setting to the other.

\subsection{A finite exact orientation certificate}
The continuum proof has a compact finite diagnostic. Restrict arrival angles to $j\pi/4$, $j=0,\ldots,7$. Before rigidity there are $4\times2\times8=64$ nonnegative coefficients. After rigidity, enumerate the $8^4=4096$ positive-outcome orientations. Exact circular-distance comparisons retain 16 families: eight common rotations and two reflections. Each retains four positive scales. The minimax proof excludes record preservation for all of them, including arbitrary scale ratios.

The matrix used in this enumeration samples the \emph{already integrated} continuous coupling $k_\Gamma(\phi_i-\phi_j+\pi)$. A discrete convolution of sampled single-leg kernels would define a different source and is not used. Exhaustive orientation enumeration is a check of the reduced finite instance; the measure-class theorem is supplied by the preceding proof.

\section{Robustness under imperfect source calibration}
\label{sec:robust}
\subsection{Calibration error and actual record disturbance}
Separate scale from shape algebraically by
\begin{equation}
\begin{aligned}
M_x&=\sum_a\int d_{x,a}>0,\\
\mu_{x,a}&=d_{x,a}/M_x,\qquad\sum_a\int\mu_{x,a}=1.
\end{aligned}
\end{equation}
This is notation, not physical row normalization. Put
\begin{equation}
\begin{aligned}
A^{ab}_{xy}&=\iint S_\Gamma\,\mu_{x,a}\mu_{y,b},\qquad
z_{xy}=\sum_{a,b}A^{ab}_{xy},\\
W^{ab}_{xy}&=M_xM_yA^{ab}_{xy},\qquad Z_{xy}=M_xM_yz_{xy}.
\end{aligned}
\label{eq:shape-scale}
\end{equation}
For $0\le\delta<1/2$, calibration requires
\begin{equation}
\begin{aligned}
\Delta(\mu)&:=\max_{x,y\in\mathcal X}\TV(P_{xy},\target_{xy})\le\delta,\\
\TV(P,Q)&=\frac12\sum_{a,b}|P^{ab}-Q^{ab}|.
\end{aligned}
\label{eq:calibration}
\end{equation}
Let $p_{xy}=\sum_bP^{+b}_{xy}$. Under the stipulated raw feedback,
\begin{equation}
p_x^f(+)=\frac{\sum_bW^{+b}_{x,f(+)}}
{\sum_bW^{+b}_{x,f(+)}+\sum_bW^{-b}_{x,f(-)}}.
\label{eq:feedback-general}
\end{equation}
Define
\begin{equation}
\begin{aligned}
\eps(\mu,M)&=\max_{x,f,y}|p_x^f(+)-p_{xy}|,\\
\epstar(\delta)&=\inf_{\Delta(\mu)\le\delta}\eps(\mu,M).
\end{aligned}
\label{eq:objective}
\end{equation}
All $4^2=16$ binary branch policies are allowed at each of the four first settings; every fixed reference $y$ is compared. The objective uses actual fixed-setting marginals, which need only satisfy $|p_{xy}-1/2|\le\delta$. It does not silently substitute ideal marginals for imperfect data.

\subsection{Quantitative rigidity near the extremum}
At an equal-setting experiment, abbreviate
\begin{equation}
\begin{aligned}
A&=A^{++}_{xx}+A^{--}_{xx},\qquad B=A^{+-}_{xx}+A^{-+}_{xx},\\
q&=\frac A{A+B},\qquad q_*:=\frac\ell{u+\ell}=\frac{1-\sech\Gamma}{2}.
\end{aligned}
\end{equation}
Write $p_a=\int\mu_{x,a}$ and $m=p_+-p_-$. The outcome-aligned probability measure shifts the negative measure by $-\pi$ before adding it to the positive one. Its first moment is
\begin{equation}
v_x=\int e^{i\phi}\mu_{x,+}(\dd\phi)-\int e^{i\phi}\mu_{x,-}(\dd\phi).
\end{equation}
Set
\begin{equation}
\begin{aligned}
\beta&=\frac{(u+\ell)\delta}{2\ell(1-q_*-\delta)},\\
H&=\min\{1,(c+1)\beta\},\qquad
L=\frac\ell{2(q_*+\delta)}.
\end{aligned}
\label{eq:beta}
\end{equation}
The denominators are positive for the stated domain.

\begin{lemma}[Saturation defect]\label{lem:defect}
Every calibrated setting satisfies
\begin{equation}
\begin{gathered}
m^2+\frac{1-|v_x|^2}{c+1}
\le\frac{uA-\ell B}{u\ell}\le\beta,\\
|v_x|\ge\sqrt{1-H},\qquad z_{xx}\ge L.
\end{gathered}
\label{eq:defect}
\end{equation}
\end{lemma}
\begin{proof}
The extrema give $A\ge\ell(p_+^2+p_-^2)$ and $B\le2up_+p_-\le u/2$. Thus $q\ge q_*$, while calibration of the same-outcome event gives $q\le q_*+\delta$. Since $A+B=B/(1-q)$,
\begin{equation}
uA-\ell B=(u+\ell)(A+B)(q-q_*)\le u\ell\beta.
\end{equation}
For the lower estimate, decompose
\begin{equation}
\begin{aligned}
uA-\ell B&=u\ell m^2+u[A-\ell(p_+^2+p_-^2)]\\
&\quad+\ell[2up_+p_--B].
\end{aligned}
\end{equation}
In outcome-aligned angle differences $t$, the excess integrands divided by $u\ell$ are $(1-\cos t)/(c+\cos t)$ for self terms and $(1-\cos t)/(c-\cos t)$ for cross terms. Both dominate $(1-\cos t)/(c+1)$. Their sum integrates to $(1-|v_x|^2)/(c+1)$. Finally, $A\ge\ell/2$ and $q\le q_*+\delta$ imply $z_{xx}=A/q\ge L$.
\end{proof}

For $H<1$ and $\vartheta_x=\arg v_x$, the aligned measure $\nu_x$ obeys the useful concentration estimate
\begin{equation}
\begin{aligned}
&\nu_x\{\operatorname{dist}(\alpha,\vartheta_x)\ge s\}\\
&\quad\le\min\left\{1,\frac{1-\sqrt{1-H}}{1-\cos s}\right\},\quad 0<s\le\pi.
\end{aligned}
\label{eq:concentration}
\end{equation}
It follows by integrating $1-\cos(\alpha-\vartheta_x)$. Approximate calibration therefore restricts both outcome masses and angular spread without imposing exact antipodal symmetry.

\subsection{Concentration-based finite-error certificate}
For a target-orthogonal setting pair, $|E_{xy}|\le2\delta$. Choose $\vartheta_x=\arg v_x$, $\vartheta_y=\arg v_y$ when nonzero and arbitrary angles otherwise. With $d_0=\vartheta_x-\vartheta_y$, the identity $S_\Gamma(c+\cos(\phi-\psi))=C_\Gamma$ gives
\begin{align}
C_\Gamma&=(c+E_{xy}\cos d_0)z_{xy}+\mathcal R,\\
\mathcal R&=\sum_{a,b}\iint S_\Gamma[\cos(\phi-\psi)-ab\cos d_0]\,\mu_{x,a}\mu_{y,b}.
\end{align}
Under the product of aligned reference measures, the bracket's squared expectation is bounded by
\begin{equation}
\E\left|e^{i(\alpha-\alpha')}-e^{id_0}\right|^2
=2(1-|v_x||v_y|)\le2H.
\end{equation}
Weighted Cauchy--Schwarz yields $|\mathcal R|\le\sqrt{2uHz_{xy}}$. Define
\begin{equation}
\begin{aligned}
a&=c-2\delta,\qquad b=\sqrt{2uH},\\
U_{\rm old}&=\min\left\{u,
\left(\frac{b+\sqrt{b^2+4aC_\Gamma}}{2a}\right)^2\right\}.
\end{aligned}
\label{eq:Uold}
\end{equation}
Solving $az_{xy}-b\sqrt{z_{xy}}-C_\Gamma\le0$ proves $z_{xy}\le U_{\rm old}$.

\begin{lemma}[Partition contrast implies record disturbance]\label{lem:odds}
If every diagonal partition shape is at least $L$ and a target-orthogonal pair has $z_{01}\le U$, define
\begin{equation}
\begin{aligned}
r&=\frac LU\frac{1-2\delta}{1+2\delta},\\
F(L,U,\delta)&=\max\left\{0,\frac{r-1}{2(r+1)}-\delta\right\}.
\end{aligned}
\label{eq:Fgeneric}
\end{equation}
Then $\eps\ge F(L,U,\delta)$ for arbitrary positive setting scales.
\end{lemma}
\begin{proof}
Use $f(+)=0$, $f(-)=1$, and let $O_x=p_x^f(+)/(1-p_x^f(+))$. Direct branch summation gives
\begin{equation}
\frac{O_0}{O_1}=
\frac{Z_{00}Z_{11}}{Z_{01}^2}
\frac{p_{00}(1-p_{11})}{(1-p_{01})p_{10}}
\ge\left[\frac LU\frac{1-2\delta}{1+2\delta}\right]^2=r^2.
\end{equation}
The scale factors cancel from $Z_{00}Z_{11}/Z_{01}^2=z_{00}z_{11}/z_{01}^2$. If $r>1$, at least one of $O_0\ge r$ or $O_1\le1/r$ holds. One adaptive probability consequently differs from $1/2$ by at least $(r-1)/[2(r+1)]$. Subtract the possible fixed-marginal error $\delta$. If $r\le1$, the clipped bound is trivial.
\end{proof}

Thus $F_{\rm old}(\delta):=F(L,U_{\rm old},\delta)$ is a universal bound. At zero error,
\begin{equation}
\begin{gathered}
H=0,\qquad L=(u+\ell)/2,\qquad U_{\rm old}=C_\Gamma/c,\\
F_{\rm old}(0)=\frac1{2\cosh(2\Gamma)}.
\end{gathered}
\end{equation}
Continuity proves a positive error neighborhood at each finite $\Gamma$. The formula has a square-root correction near zero, but this will be improved below. Its unclipped zero at $\Gamma=0.4$ is approximately $0.00575528$; this is the point where this certificate stops excluding preservation, not a demonstrated feasibility threshold.

\subsection{Harmonic-mean refinement of the robustness bound}
Both outcome masses are positive when $\delta<1/2$. Define conditional first moments and an unaligned moment by
\begin{equation}
\begin{aligned}
r_{x,a}&=\frac{\int e^{i\phi}\mu_{x,a}(\dd\phi)}{p_{x,a}},\\
Q_x&=\sum_a p_{x,a}(1-|r_{x,a}|^2),\qquad t_x=\sum_a p_{x,a}r_{x,a}.
\end{aligned}
\end{equation}
Jensen's inequality and \cref{eq:defect} give $Q_x\le1-|v_x|^2\le H$. Also
\begin{equation}
\begin{aligned}
|t_x|^2+|v_x|^2
&=2\bigl(|p_{x,+}r_{x,+}|^2+|p_{x,-}r_{x,-}|^2\bigr)\\
&\le1+m_x^2.
\end{aligned}
\end{equation}
Combining with the joint defect inequality gives $|t_x|^2\le H$ and
$\sum_a p_{x,a}^2\le[1+\min\{\beta,1\}]/2$.

\begin{lemma}[Reciprocal-kernel Jensen gap]\label{lem:harmonic}
For independent conditional arrival angles, let $T=S_\Gamma(\phi,\psi)$ and $z=\cos(\phi-\psi)$. Then
\begin{equation}
\begin{aligned}
0\le\E(1/T)-\frac1{\E T}
&\le\frac{\Var(z)}{C_\Gamma(c-1)}\\
&\le\frac{2-|r_{x,a}|^2-|r_{y,b}|^2}{C_\Gamma(c-1)}.
\end{aligned}
\label{eq:jensen-gap}
\end{equation}
\end{lemma}
\begin{proof}
Put $Y=c+z$, $\bar Y=\E Y$, and $h_Y=1/\E(1/Y)$. An exact identity is
\begin{equation}
\bar Y-h_Y=
\frac{\E[(Y-\bar Y)^2/Y]}{\bar Y\E(1/Y)}\le\frac{\Var(Y)}{c-1}.
\end{equation}
Use $Y\ge c-1$ and $\bar Y\E(1/Y)\ge1$, then divide by $C_\Gamma$. The variance of a real part is no greater than its complex variance, so
$\Var(z)\le1-|r_{x,a}|^2|r_{y,b}|^2\le(1-|r_{x,a}|^2)+(1-|r_{y,b}|^2)$.
\end{proof}

Multiply by the conditional reference masses and sum outcomes:
\begin{equation}
\frac{c+\operatorname{Re}(t_x\overline{t_y})}{C_\Gamma}
\le\sum_{a,b}\frac{p_{x,a}^2p_{y,b}^2}{A^{ab}_{xy}}
+\frac{Q_x+Q_y}{C_\Gamma(c-1)}.
\label{eq:harmonic-sum}
\end{equation}
Both squared masses are necessary since $A^{ab}_{xy}=p_{x,a}p_{y,b}\E T$. For an orthogonal target pair, $P^{ab}_{xy}\ge1/4-\delta$. When $\delta<1/4$,
\begin{equation}
\frac{c-H}{C_\Gamma}
\le\frac{[1+\min\{\beta,1\}]^2}{z_{xy}(1-4\delta)}+
\frac{2H}{C_\Gamma(c-1)}.
\end{equation}
Define
\begin{equation}
\begin{aligned}
D_{\rm lin}&=c-H-\frac{2H}{c-1},\\
U_{\rm lin}&=\begin{cases}
\min\left\{u,\dfrac{C_\Gamma[1+\min\{\beta,1\}]^2}{(1-4\delta)D_{\rm lin}}\right\},\\[-2pt]
\hfill \delta<1/4,\ D_{\rm lin}>0,\\
u,\hfill\text{otherwise}.
\end{cases}
\end{aligned}
\label{eq:Ulin}
\end{equation}
Thus $z_{xy}\le U_{\rm lin}$. The combined certificate is
\begin{equation}
F_{\rm new}(\delta)=F\bigl(L,\min\{U_{\rm old},U_{\rm lin}\},\delta\bigr).
\label{eq:Fnew}
\end{equation}

\begin{theorem}[Robust obstruction]\label{thm:robust}
For (D1)--(D4), fixed finite $\Gamma>0$, and the calibration condition \cref{eq:calibration},
\begin{equation}
\epstar(\delta)\ge F_{\rm new}(\delta)\ge F_{\rm old}(\delta).
\end{equation}
For sufficiently small error the improved bound has a linear, rather than square-root, loss from its exact endpoint.
\end{theorem}
\begin{proof}
Apply \cref{lem:odds} using either proved cross-weight bound, and take their minimum. The explicit expansion given next establishes the linear claim.
\end{proof}

\subsection{Sharp asymptotic order at every fixed positive width}
Put $\rho=\coth^2\Gamma$ and $B_\Gamma=2c^2/(c^2-1)$. Expanding the linear certificate gives
\begin{equation}
\begin{aligned}
F_{\rm lin}(\delta)&=\epszero-K_\Gamma\delta+O_\Gamma(\delta^2),\\
K_\Gamma&=1+\frac\rho{(\rho+1)^2}
\biggl[\frac1{q_*}+2B_\Gamma+8\\
&\hspace{5em}+\frac{B_\Gamma(c+1)^2}{c(c-1)}\biggr],\\
\epszero&=\frac1{2\cosh(2\Gamma)}.
\end{aligned}
\label{eq:Kgamma}
\end{equation}
To supply a constructive comparison, choose endpoint responses smoothed by a wrapped-Cauchy density:
\begin{equation}
\mu_{x,a}(\dd\phi)=\frac12k_\sigma(\phi-\theta(x,a))\dd\phi,\qquad\sigma>0.
\end{equation}
The source coupling remains $S_\Gamma$, but exact response convolution produces table width $G=\Gamma+2\sigma$. Equal scales give
\begin{equation}
\delta_\sigma=\frac12(\sech\Gamma-\sech G),\qquad
\eps_\sigma=\frac1{2\cosh(2G)}.
\end{equation}
For $v=\sech\Gamma$ and $v-2\delta>0$,
\begin{equation}
\begin{aligned}
\epstar(\delta)&\le G_{\rm smooth}(\delta)
=\frac{(v-2\delta)^2}{2[2-(v-2\delta)^2]}\\
&=\epszero-a_\Gamma\delta+O_\Gamma(\delta^2),\\
a_\Gamma&=\frac{4v}{(2-v^2)^2}>0.
\end{aligned}
\label{eq:smooth-upper}
\end{equation}
The atomic endpoint is allowed at $\delta=0$.

\begin{corollary}[Linear small-error order]\label{cor:order}
For each fixed finite $\Gamma>0$,
\begin{equation}
\epszero-\epstar(\delta)=\Theta_\Gamma(\delta),\qquad\delta\downarrow0.
\end{equation}
More explicitly, the liminf of $[\epszero-\epstar(\delta)]/\delta$ is at least $a_\Gamma>0$, and its limsup is at most $K_\Gamma<\infty$.
\end{corollary}
\begin{proof}
Sandwich the infimum between \cref{eq:smooth-upper} and the linear lower certificate. No attainment or differentiability of the full optimum is assumed.
\end{proof}
At $\Gamma=0.4$, $K_\Gamma\simeq83.34211690$ and $a_\Gamma\simeq2.82539364$. The original lower certificate instead has a strictly negative $\sqrt\delta$ leading correction. Since the improved valid lower bound is larger for sufficiently small $\delta$, the older bound cannot be attained or approached there. This proves its non-tightness analytically.

\subsection{Finite-error comparison families}
A positive five-atom response, developed in \cref{sec:sharp}, improves on smoothing near the exact endpoint. \Cref{tab:bounds} reports direct contractions of that response, providing constructive upper bounds on the finite-error optimum. \Cref{fig:bounds} resolves the near-zero regime, while \cref{fig:bounds-finite} shows the broader finite-error comparison.
\begin{table}[tbp]
\centering\small
\setlength{\tabcolsep}{3pt}
\caption{Bounds on $\epstar(\delta)$ at $\Gamma=0.4$. The two lower bounds are $F_{\rm old}$ (concentration) and $F_{\rm new}$ (combined); the upper bound is attained by the five-atom family. Values are rounded to eight decimal places.}
\label{tab:bounds}
\begin{tabular}{@{}rrrr@{}}\toprule
$\delta$ & $F_{\rm old}$ & $F_{\rm new}$ & Upper bound\\\midrule
$0$&0.37384996&0.37384996&0.37384996\\
$10^{-6}$&0.37099062&0.37376657&0.37384061\\
$10^{-5}$&0.36458849&0.37301162&0.37375648\\
$10^{-4}$&0.34232077&0.36499658&0.37291775\\
$10^{-3}$&0.25146332&0.25146332&0.36477274\\
$5\times10^{-3}$&0.03076255&0.03076255&0.33269036\\\bottomrule
\end{tabular}
\end{table}
\begin{figure}[tbp]
\centering\includegraphics[width=\columnwidth]{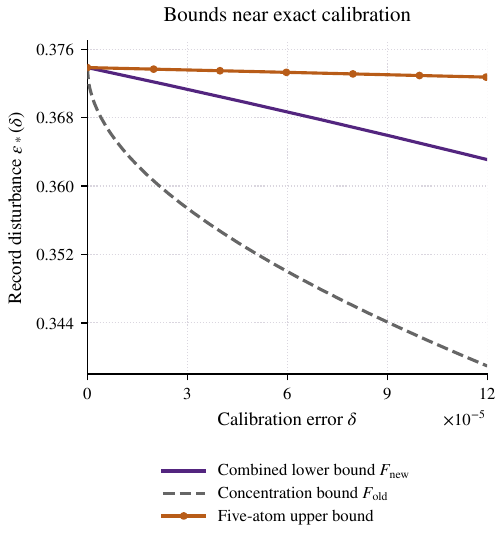}
\caption{Bounds near exact calibration at $\Gamma=0.4$. The combined lower bound $F_{\rm new}$ uses the harmonic-mean estimate to improve the concentration bound $F_{\rm old}$. The five-atom curve is an achieved upper bound, evaluated at its actual calibration error. The optimum $\epstar(\delta)$ lies between the combined lower and constructive upper curves.}
\label{fig:bounds}
\end{figure}
\begin{figure}[tbp]
\centering\includegraphics[width=\columnwidth]{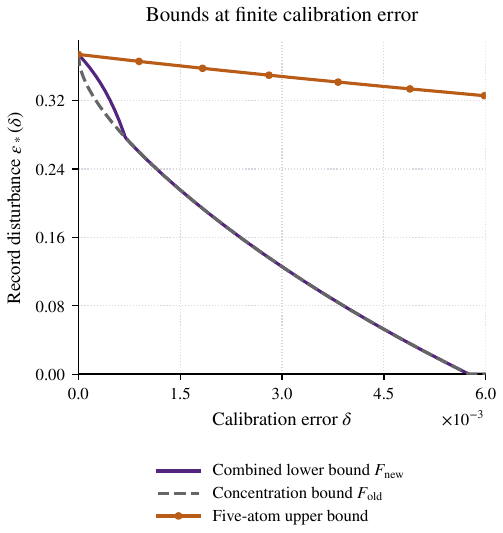}
\caption{Bounds over a wider calibration-error range at $\Gamma=0.4$, with the same curve definitions as \cref{fig:bounds}. Farther from zero, the combined and concentration lower bounds coincide; the grey dashes overlay the purple curve there. The five-atom curve supplies a constructive upper bound throughout the plotted range.}
\label{fig:bounds-finite}
\end{figure}

For exact preservation at a larger error, consider constant-total responses. The cosine family
\begin{equation}
\mu^{\cos}_{x,a}(\dd\phi)=\frac{1+a\cos(\phi-x)}{4\pi}\dd\phi
\end{equation}
has $z_{xy}=1/(2\pi)$, uniform marginals, and correlation $E(d)=-\tfrac12e^{-\Gamma}\cos d$. It therefore has $\eps=0$ and
$\delta_{\cos}=\tfrac12(\sech\Gamma-\tfrac12e^{-\Gamma})$, approximately $0.2949237144$ at the benchmark.

A better preserving example uses half-circle outcomes:
\begin{equation}
\mu^{\rm half}_{x,a}(\dd\phi)=\frac{1+a\sgn\cos(\phi-x)}{4\pi}\dd\phi.
\label{eq:half-circle}
\end{equation}
Their total is independent of $x$, and Fourier integration gives
\begin{equation}
E_{\rm half}(d)=-\frac8{\pi^2}\sum_{\substack{n\ge1\\n\ \mathrm{odd}}}
\frac{e^{-n\Gamma}\cos(nd)}{n^2}.
\label{eq:half-series}
\end{equation}
All first-record marginals are preserved. On the four-setting alphabet at $\Gamma=0.4$, its error is approximately $0.1743750480$. Calibration also gives $S_{\rm CHSH}\ge2\sqrt2\sech\Gamma-8\delta$, since each of the four signed correlation terms changes by at most $2\delta$. At $\Gamma=0.4$ and $\delta=0.005$, this remains above $2.5763$. Thus the robust obstruction has a calibration regime that still guarantees Bell violation, although its proof does not use that violation.

These constant-total models have a setting-independent hidden arrival distribution and local outcome responses, and are Bell-local. Their performance does not locate the least error permitting preservation.

The lower bounds require only two orthogonal settings and their diagonals, so extend to larger calibrated alphabets. The smoothing and constant-total formulas hold for continuous angles as well. The five-atom sharp coefficient below, however, is tied to the specified finite alphabet; extending its calibration test to every angle changes the optimization problem.

\section{The sharp coefficient at the benchmark width}
\label{sec:sharp}
The linear order leaves a substantial coefficient gap. At $\Gamma=0.4$, an explicit five-atom perturbation already improves the smoothing coefficient from approximately $2.8254$ to more than $9.3508$, but a restricted numerical search alone cannot prove optimality over arbitrary detector measures. We close that gap by a uniform concentration argument, a first-order symmetry reduction, and matching positive-measure and polynomial certificates.

\begin{theorem}[Sharp benchmark expansion]\label{thm:sharp}
For (D1)--(D4), the alphabet \cref{eq:alphabet}, and fixed $\Gamma=0.4$,
\begin{equation}
\begin{aligned}
\epstar(\delta)&=\epszero-\Csharp\delta+O(\delta^2),\\
\epszero&=\frac1{2\cosh(0.8)},\qquad\delta\downarrow0.
\end{aligned}
\label{eq:sharp}
\end{equation}
where $\Csharp=\lambda_0+\lambda_1$ is defined by the unique root in the explicit interval box of \cref{app:certificate}, and
\begin{equation}
\begin{gathered}
9.35109177105700958345410586084\\
{}<\Csharp<{}\\
9.35109177105700958345410586086.
\end{gathered}
\label{eq:sharp-interval}
\end{equation}
The lower estimate covers arbitrary finite nonnegative terminal measures, asymmetric outcomes, arbitrary positive setting scales, and error-dependent families. A nonnegative five-atom family attains the same coefficient.
\end{theorem}
The proof combines a uniform reduction over detector measures with matching lower and constructive upper bounds. Directed interval arithmetic certifies the root and continuous polynomial inequalities required for the sharp coefficient.

\subsection{A uniform gauge for nearly calibrated measures}
Let $\mu^0_{x,a}=\frac12\delta_{\theta(x,a)}$. For a signed measure $\xi$ on $\T$, define
\begin{equation}
\norm{\xi}_{-2}=\sup_{\norm{f}_{C^2}\le1}\left|\int f\dd\xi\right|,
\end{equation}
where the norm bounds $f$ and its first two derivatives. Constants below may depend on the fixed width and finite alphabet, but not on the detector family.

\begin{lemma}[Uniform rotation gauge]\label{lem:gauge}
For every sufficiently small calibrated family, a common physical rotation and, if needed, a common reflection give
\begin{equation}
\max_{x,a}\norm{\mu_{x,a}-\mu^0_{x,a}}_{-2}=O_\Gamma(\delta).
\label{eq:gauge}
\end{equation}
Both transformations preserve every fixed-setting and feedback probability exactly.
\end{lemma}
\begin{proof}
The defect bound initially gives $m_x^2=O(\delta)$ and $1-|v_x|^2=O(\delta)$. Diagonal outcome balance improves the mass estimate. Write
$A_{\rm exc}=A-\ell(p_+^2+p_-^2)\ge0$. The defect decomposition gives $A_{\rm exc}\le\ell\beta$. Since the target probabilities for $++$ and $--$ agree, calibration implies
\begin{equation}
\begin{aligned}
|A^{++}_{xx}-A^{--}_{xx}|&\le2z_{xx}\delta,\\
|m_x|&\le\frac{2u\delta}{\ell}+\beta=O(\delta).
\end{aligned}
\label{eq:mass-Odelta}
\end{equation}
The second inequality uses $A^{++}_{xx}-A^{--}_{xx}=\ell m_x$ plus a difference of nonnegative self-excesses whose sum is $A_{\rm exc}$, together with $z_{xx}\le u$.

Let $\vartheta_x=\arg v_x$. The aligned probability measure has mean squared chord $2(1-|v_x|)=O(\delta)$ about $\vartheta_x$. Because each outcome mass is $1/2+O(\delta)$, each conditional outcome concentrates to the same order around $\vartheta_x+\pi(1-a)/2$.

Lipschitz continuity of the strictly positive smooth kernel makes the normalized source table $O(\sqrt\delta)$ close to the ideal endpoint table at these orientations. Consequently their cosine Gram matrix differs from the canonical one by $O(\sqrt\delta)$. Rotate $\vartheta_0$ to zero. The orthogonal correlation puts $\vartheta_{\pi/2}$ within $O(\sqrt\delta)$ of either $\pi/2$ or $-\pi/2$; choose a common reflection. Correlations with these two independent axes locate the remaining two directions within $O(\sqrt\delta)$ of the canonical ones.

For each conditional outcome, express its angle relative to the canonical endpoint as $\alpha$. Let $u_{x,a}=\E\sin\alpha$ and $q_{x,a}=\E(1-\cos\alpha)$. At this stage $u_{x,a}=O(\sqrt\delta)$ and $q_{x,a}=O(\delta)$. The periodic Taylor inequality
\begin{equation}
|f(\theta+\alpha)-f(\theta)-f'(\theta)\sin\alpha|
\le K\norm{f}_{C^2}(1-\cos\alpha)
\label{eq:periodic-taylor}
\end{equation}
holds on the entire circle. An explicit admissible choice is $K=\pi^2/2$: for the representative $|\alpha|\le\pi$, Taylor's remainder and $|\alpha-\sin\alpha|\le\alpha^2/2$ give an upper bound $\norm{f}_{C^2}\alpha^2$, and $\alpha^2\le(\pi^2/2)(1-\cos\alpha)$. Thus
\begin{equation}
P-P^0=\mathsf J u+O_\Gamma(\delta),
\label{eq:Jacobian}
\end{equation}
where $\mathsf J$ is the normalized source-table Jacobian for the eight independent endpoint shifts. Products of the initial $O(\sqrt\delta)$ shifts are absorbed into the remainder.

The kernel of $\mathsf J$ is precisely a common rotation. For an orthogonal setting pair, a zero normalized derivative requires $ab(u_{x,a}-u_{y,b})$ to equal one constant $k$ for all outcomes. The four equations are
\begin{align}
u_{x,+}-u_{y,+}&=k,&u_{x,+}-u_{y,-}&=-k,\\
u_{x,-}-u_{y,+}&=-k,&u_{x,-}-u_{y,-}&=k.
\end{align}
Their consistency forces $k=0$, so all four shifts agree. Each orthogonal pair therefore has one common shift. An oblique cross-pair has a nonzero angular derivative and makes the two common shifts equal. Hence $\mathsf J$ has rank seven. On its finite-dimensional complement, \cref{eq:Jacobian} and calibration give $u_{x,a}=u_{\rm com}+O(\delta)$.

Rotate once more by $-u_{\rm com}=O(\sqrt\delta)$. The remaining sine moments are $O(\delta)$, their squared chords remain $O(\delta)$, and the outcome-mass errors are $O(\delta)$. More explicitly, if $p_{x,a}$ is the outcome mass and $u_{x,a},q_{x,a}$ are the resulting conditional moments, then
\begin{equation}
\begin{aligned}
&\norm{\mu_{x,a}-\tfrac12\delta_{\theta(x,a)}}_{-2}\\
&\quad\le |p_{x,a}-\tfrac12|+p_{x,a}|u_{x,a}|+
\frac{\pi^2}{2}p_{x,a}q_{x,a}\\
&\quad=O(\delta).
\end{aligned}
\label{eq:gauge-explicit-bound}
\end{equation}
This proves \cref{eq:gauge} without assuming a bounded density or a differentiable detector family.
\end{proof}

Uniformity matters here. If $\xi_{x,a}=\mu_{x,a}-\mu^0_{x,a}$ after this gauge, then
\begin{equation}
\iint S_\Gamma(\phi,\psi)\,\xi_{x,a}(\dd\phi)\xi_{y,b}(\dd\psi)=O_\Gamma(\delta^2).
\label{eq:bilinear-remainder}
\end{equation}
The mixed derivatives through order $(2,2)$ of $S_\Gamma$ are bounded at fixed positive width. Applying the dual norm successively in each variable proves the estimate. Since every $z_{xy}\ge\ell>0$, normalized probabilities have uniform quadratic remainders as well. This covers narrow error-dependent distributions with angular width $O(\sqrt\delta)$, not only differentiable finite-parameter families.

\subsection{Validity of first-order symmetry reduction}
Write $x^\perp=x+\pi/2$ modulo $\pi$, with the corresponding outcome relabeling when necessary, and define
\begin{equation}
\sigma_x=\tfrac12(\log M_x-\log M_{x^\perp}).
\end{equation}
Only competitors with $\eps\le\epszero$ need be considered for a lower bound. For their orthogonal feedback comparisons, the log odds are
\begin{equation}
\begin{aligned}
\log O_x^+&=\log\rho+2\sigma_x+O(\delta),\\
\log O_x^-&=-\log\rho-2\sigma_x+O(\delta).
\end{aligned}
\end{equation}
The disturbance restriction and $p_{xx}=1/2+O(\delta)$ bound $\sigma_x$ above by $O(\delta)$. Repeating at $x^\perp$ bounds it below, giving $\sigma_x=O(\delta)$. The common scales of the two orthogonal pairs need not approach each other; only their within-pair ratios enter the witnesses used below.

Let $\mathcal G$ be the dihedral group of the eight endpoints $k\pi/4$, acting by physical rotation through $\pi/4$, reflection, and the induced setting/output relabelings. The ideal family and target tables are invariant. Averaging the normalized shapes over $\mathcal G$ produces a positive family of the form
\begin{equation}
\bar\mu_{x,a}=\tfrac12\tau_{\theta(x,a)}\nu,
\label{eq:symfamily}
\end{equation}
where $\nu$ is an even probability measure on $\T$ and $\tau$ denotes translation. The group average of the vector $\sigma$ vanishes. Equivariance, convexity of total variation, and the uniform remainder imply
\begin{equation}
\Delta(\bar\mu)\le\Delta(\mu)+O(\delta^2).
\label{eq:symcal}
\end{equation}
This is a first-order statement. Independently averaging responses at two detectors does not exactly average their finite-error source tables.

To control the objective without assuming its differentiability, use $f_x^+=(+\mapsto x,-\mapsto x^\perp)$ and its reverse $f_x^-$. Define eight signed witnesses and their average:
\begin{equation}
\begin{aligned}
C_x^+&=p_x^{f_x^+}(+)-p_{xx},\\
C_x^-&=p_{xx}-p_x^{f_x^-}(+),\\
w(\mu,\sigma)&=\tfrac18\sum_{x\in\mathcal X}(C_x^++C_x^-).
\end{aligned}
\end{equation}
Each signed witness is bounded above by $\eps$, so $\eps\ge w$. The fixed marginal cancels exactly from their sum:
\begin{equation}
w=\tfrac18\sum_x\left[p_x^{f_x^+}(+)-p_x^{f_x^-}(+)\right].
\label{eq:w-cancel}
\end{equation}
An explicit expansion verifies the symmetry step. Define
\begin{equation}
J_{\rm d}=\frac{C_\Gamma c}{c^2-1},\qquad J_{\rm o}=\frac{C_\Gamma}{c},\qquad
h_\rho=\frac{\rho}{(1+\rho)^2}.
\end{equation}
These are the ideal diagonal and orthogonal partition shapes and the logistic derivative at either ideal feedback odds. Let $s_{xy}^a=\sum_b A^{ab}_{xy}$. The two exact log odds are
\begin{align}
\log O_x^+&=2\sigma_x+\log s_{xx}^+-\log s_{x,x^\perp}^-,\\
\log O_x^-&=-2\sigma_x+\log s_{x,x^\perp}^+-\log s_{xx}^-.
\end{align}
At the ideal family, $s_{xx}^a=J_{\rm d}/2$ and $s_{x,x^\perp}^a=J_{\rm o}/2$. Expanding the logarithms and the logistic map, with $\sigma_x=O(\delta)$ and the uniform shape estimates, yields
\begin{equation}
\begin{aligned}
w(\mu,\sigma)&=\epszero+\frac{h_\rho}{4}\sum_x
\biggl[\frac{z_{xx}-J_{\rm d}}{J_{\rm d}}\\
&\qquad-\frac{z_{x,x^\perp}-J_{\rm o}}{J_{\rm o}}\biggr]+O(\delta^2).
\end{aligned}
\label{eq:w-explicit-expansion}
\end{equation}
The first-order scale term is $(h_\rho/2)\sum_x\sigma_x=0$, since $\sigma_{x^\perp}=-\sigma_x$. The outcome-specific branch perturbations combine into the total partition perturbations displayed here. No equality between the baseline scales of different orthogonal pairs is used.

The linear part of the bracketed sum is invariant under the finite group: the group permutes diagonal pairs and orthogonal pairs, and output relabeling leaves each total partition unchanged. Replacing the shapes by their group average therefore preserves this linear part. The bilinear estimate \cref{eq:bilinear-remainder} bounds the remaining difference uniformly, giving
\begin{equation}
w(\mu,\sigma)=w(\bar\mu,0)+O(\delta^2).
\end{equation}
Similarly, linearizing the normalized source tables gives
\begin{equation}
P(\bar\mu)-P^0=\frac1{|\mathcal G|}\sum_{g\in\mathcal G}
 g\mathbin{\cdot}[P(\mu)-P^0]+O(\delta^2),
\end{equation}
where $g\mathbin{\cdot}$ includes the induced table relabeling. Convexity of total variation proves \cref{eq:symcal}. This supplies an explicit justification of first-order averaging; exact finite-error equality is not asserted.

For \cref{eq:symfamily}, fixed marginals are exactly $1/2$, and the strict ideal partition ordering at separations $0,\pi/4,\pi/2$ persists nearby. The largest equal-scale feedback disturbance is precisely the diagonal-to-orthogonal one. Hence
\begin{equation}
\eps(\mu,M)\ge\eps(\bar\mu,\bm1)-O(\delta^2).
\label{eq:full-reduction}
\end{equation}
No symmetry assumption has been imposed on the original detector.

\subsection{Reduction to a positive moment problem}
For a symmetric family set
\begin{equation}
\begin{aligned}
K_\nu(d)&=\iint S_\Gamma(d+\alpha-\beta)\nu(\dd\alpha)\nu(\dd\beta),\\
S_\Gamma(d)&=\frac{C_\Gamma}{c+\cos d}.
\end{aligned}
\end{equation}
Its linear change from the ideal value is
\begin{equation}
\int[S_\Gamma(d+\alpha)+S_\Gamma(d-\alpha)-2S_\Gamma(d)]\nu(\dd\alpha),
\end{equation}
with a uniform $O(\delta^2)$ remainder. Let $z=\cos\alpha$, $h=c^2$, and $b=h-1/2$. Differentiating the diagonal correlation error, oblique correlation error, and extreme partition ratio yields
\begin{align}
a_0(z)&=\frac{h-1}{c}\frac{1-z}{h-z^2},\label{eq:a0}\\
a_1(z)&=\frac{b}{\sqrt2c}\frac{(1-z)(b-z-z^2)}{b^2-z^2+z^4},\label{eq:a1}\\
g(z)&=\frac{2c^2(c^2-1)(1-z^2)}{(2c^2-1)(c^2-z^2)(c^2-1+z^2)}.
\label{eq:g}
\end{align}
Integrals below use the pushforward of $\nu$ under cosine. Then
\begin{align}
\Delta(\bar\mu)&=\max\left\{\int a_0\dd\nu,\left|\int a_1\dd\nu\right|\right\}+O(\delta^2),\label{eq:moment-error}\\
\eps(\bar\mu,\bm1)&=\epszero-\int g\dd\nu+O(\delta^2).
\label{eq:moment-objective}
\end{align}
Here $a_0\ge0$. The orthogonal correlation vanishes exactly; the opposite oblique orientation supplies the same absolute error. \Cref{app:derivatives} gives a direct derivative recipe for these expressions.

Introduce the continuous ratios, extended at $z=1$,
\begin{align}
A(z)=\frac{a_1(z)}{a_0(z)}
&=\frac{b}{\sqrt2(h-1)}\frac{(h-z^2)(b-z-z^2)}{b^2-z^2+z^4},\\
G(z)=\frac{g(z)}{a_0(z)}
&=\frac{2c^3(1+z)}{(2h-1)(h-1+z^2)}.
\label{eq:AG}
\end{align}
An upper bound on the coefficient follows from the positive-measure optimization
\begin{equation}
\begin{aligned}
\sup_{\eta\ge0}\biggl\{\int G\dd\eta:\ &\int\dd\eta\le1,\quad\int A\dd\eta\le1,\\
&z\in[-1,1]\biggr\}.
\end{aligned}
\label{eq:LP}
\end{equation}
The original absolute oblique constraint also gives $\int A\dd\eta\ge-1$. Dropping that lower constraint makes \cref{eq:LP} a relaxation, which is sufficient for the upper bound on possible improvement. The matching construction below satisfies $\int A\dd\eta=1$, so the relaxation loses nothing at the benchmark optimum.

It suffices to find $\lambda_0,\lambda_1>0$ with
\begin{equation}
G(z)\le\lambda_0+\lambda_1A(z)\quad(-1\le z\le1)
\label{eq:dual}
\end{equation}
and a positive measure supported at equality points that saturates both constraints. Their matching values give the coefficient without requiring an abstract strong-duality theorem.

\subsection{Continuous polynomial certification of the dual bound}
Define
\begin{equation}
\begin{aligned}
D(z)&=(h-1+z^2)(b^2-z^2+z^4),\\
\mathcal P(z)&=\lambda_0D(z)+\lambda_1E(z)-N(z),\\
E(z)&=\frac{b}{\sqrt2(h-1)}(h-z^2)(b-z-z^2)\\
&\qquad\times(h-1+z^2),\\
N(z)&=\frac{2c^3}{2h-1}(1+z)(b^2-z^2+z^4).
\end{aligned}
\label{eq:poly}
\end{equation}
The denominator is strictly positive because
$b^2-z^2+z^4=(z^2-1/2)^2+h(h-1)>0$. Thus
$\mathcal P/D=\lambda_0+\lambda_1A-G$.

At $\Gamma=0.4$, the four unknowns are defined by the double-contact equations
\begin{equation}
\mathcal P(r_1)=\mathcal P'(r_1)=\mathcal P(r_2)=\mathcal P'(r_2)=0.
\label{eq:root-system}
\end{equation}
The explicit interval box in \cref{app:certificate} contains one unique root, approximately
\begin{equation}
\begin{aligned}
\lambda_0&=8.6124493990175556,\\
\lambda_1&=0.7386423720394540,\\
r_1&=0.7187943057772323,\\
r_2&=0.1301153680796272.
\end{aligned}
\end{equation}
The interval verification uses 90-decimal-digit interval arithmetic and a box radius $10^{-45}$. Its preconditioned fixed-point map lies strictly within the box and has derivative infinity norm less than $3.614\times10^{-43}$. This validates existence and uniqueness independently of a floating-point root solver's stopping criterion.

The distinct exact double roots imply exact factorization
\begin{equation}
\mathcal P(z)=(z-r_1)^2(z-r_2)^2(q_2z^2+q_1z+q_0).
\label{eq:factorization}
\end{equation}
Interval polynomial division gives
\begin{equation}
\begin{aligned}
q_0&\simeq8.90059792472045,\\
q_1&\simeq18.31788401745949,\\
q_2&\simeq10.68259946195851.
\end{aligned}
\end{equation}
with $q_2>0$ and $q_1^2-4q_0q_2<-44.7812$. Therefore the quadratic is positive on the real line, and $\mathcal P(z)\ge0$ everywhere. The double-root equations, not merely a small numerical remainder, establish exact divisibility.

Consequently,
\begin{equation}
g(z)\le\lambda_0a_0(z)+\lambda_1a_1(z),\qquad -1\le z\le1.
\label{eq:moment-dual}
\end{equation}
Integrate, use the positive multipliers and \cref{eq:symcal,eq:moment-error,eq:full-reduction}, and obtain the full-class lower estimate
\begin{equation}
\eps\ge\epszero-(\lambda_0+\lambda_1)\delta-O(\delta^2).
\label{eq:sharp-lower}
\end{equation}

\Cref{fig:dual-certificate} shows the continuously certified slack and its two equality points.
\begin{figure}[tbp]
\centering\includegraphics[width=\columnwidth]{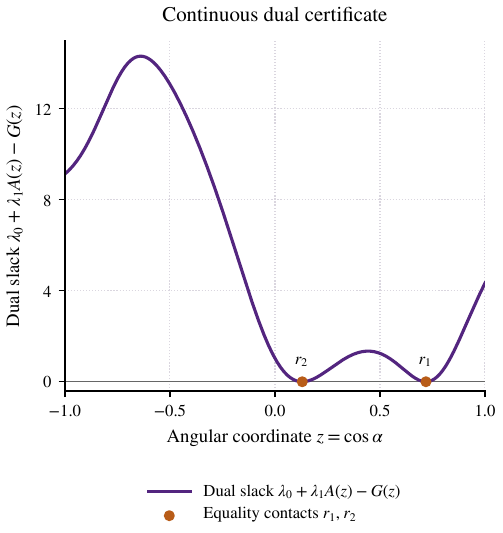}
\caption{Continuous dual certificate at $\Gamma=0.4$. The slack $\lambda_0+\lambda_1A(z)-G(z)$ is nonnegative for every $z\in[-1,1]$, as proved by the exact double-root factorization and interval bounds on the quadratic factor. The marked contacts $r_1$ and $r_2$ attain equality.}
\label{fig:dual-certificate}
\end{figure}

\subsection{Attainment of the sharp coefficient by a five-atom detector}
At the contacts, $A(r_1)\simeq-5.2451097961$ and $A(r_2)\simeq3.9113293492$. Define
\begin{equation}
\eta_1=\frac{1-A(r_2)}{A(r_1)-A(r_2)},\qquad\eta_2=1-\eta_1.
\end{equation}
Both are interval-certified positive, approximately $0.3179543164$ and $0.6820456836$. They satisfy $\eta_1+\eta_2=1$ and $\eta_1A(r_1)+\eta_2A(r_2)=1$. Contact equality gives
\begin{equation}
\eta_1G(r_1)+\eta_2G(r_2)=\lambda_0+\lambda_1=\Csharp.
\end{equation}
Let $\alpha_j=\arccos r_j$ and $w_j=\eta_j/a_0(r_j)$. Numerically,
\begin{equation}
\begin{aligned}
\alpha_1&\simeq0.7687298222010835,\\
\alpha_2&\simeq1.440310990731049,\\
w_1&\simeq4.724084942587050,\\
w_2&\simeq5.786538673414912.
\end{aligned}
\end{equation}
For $0\le t\le1/(w_1+w_2)$, the response
\begin{equation}
\begin{aligned}
\nu_t&=[1-(w_1+w_2)t]\delta_0\\
&\quad+\sum_{j=1}^2\frac{w_jt}{2}(\delta_{\alpha_j}+\delta_{-\alpha_j}),\\
\mu_{x,a}&=\tfrac12\tau_{\theta(x,a)}\nu_t,\qquad M_x=1.
\end{aligned}
\label{eq:fiveatom}
\end{equation}
is nonnegative and belongs to the original measure class. Its first-order moments give
\begin{equation}
\int a_0\dd\nu_t=\int a_1\dd\nu_t=t,\qquad
\int g\dd\nu_t=\Csharp t.
\end{equation}
Therefore
\begin{equation}
\Delta(t)=t+O(t^2),\qquad\eps(t)=\epszero-\Csharp t+O(t^2).
\label{eq:attainment}
\end{equation}
The finitely many error branches are analytic near zero, and the active branches have positive first derivative. Their maximum is locally increasing, permitting $t=\delta+O(\delta^2)$. Equivalently, take $t=\delta-K\delta^2$ with fixed sufficiently large $K$ to ensure feasibility. This upper bound matches \cref{eq:sharp-lower}, proving \cref{thm:sharp}.

\Cref{fig:coefficient} displays the approach of the nonlinear five-atom family to the certified coefficient.
\begin{figure}[tbp]
\centering\includegraphics[width=\columnwidth]{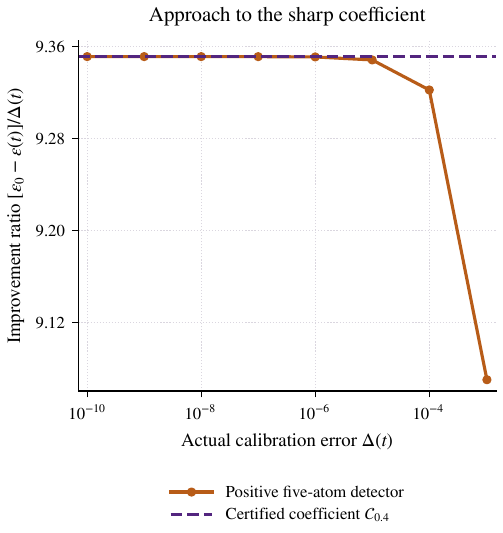}
\caption{Attainment of the sharp coefficient at $\Gamma=0.4$. High-precision positive-history sums give the five-atom improvement ratio $[\epszero-\eps(t)]/\Delta(t)$, which approaches the certified value $\Csharp\simeq9.351091771057010$ as the actual calibration error decreases (towards the left on the logarithmic axis).}
\label{fig:coefficient}
\end{figure}

The theorem establishes the right derivative $\epstar'(0+)=-\Csharp$ for the specified four-setting benchmark. Direct finite-error evaluations illustrate attainment: at $t=10^{-6}$ the five-atom family has $\Delta(t)\simeq1.00002882885\times10^{-6}$ and improvement-to-error ratio $9.35080031713$; at $t=10^{-10}$ the ratio is $9.35109174191$. The finite-error optimum, width dependence, and continuous-setting extension define the subsequent optimization problems.

\section{Verification and reproducibility}
\label{sec:verification}
The verification combines analytical measure-class proofs, a directed interval certificate for the benchmark coefficient, and numerical checks of the implemented formulas. The interval computation certifies the finite-dimensional root and continuous polynomial inequalities. Direct outcome contractions, quadrature, orientation enumeration, and a separately implemented asymmetric tangent program check the source laws, detector formulas, and first-order reduction.

The accompanying \texttt{verify\_study.py} regenerates the numerical tables, vector figures, verification report, and coefficient certificate. The additional programs \texttt{audit\_reduction.py} and \texttt{audit\_tangent.py} implement the checks described in \cref{app:audit}. The computations use NumPy, SciPy, mpmath, and Matplotlib; numerical examples use seed 20260912.

\Cref{tab:verification} summarizes the verification scope and results.

\begin{table}[tbp]
\centering
\footnotesize
\setlength{\tabcolsep}{3pt}
\renewcommand{\arraystretch}{1.08}

\caption{Verification of the source, detector constructions,
and sharp coefficient. The continuous certificate supplies
the validated numerical component of the proof.}
\label{tab:verification}

\begin{tabular}{@{}ll@{}}
\toprule
Calculation & Scope and result \\
\midrule

\shortstack[l]{Continuous\\coefficient\\certificate}
&
\shortstack[l]{Interval inclusion and contraction\\
of the four-variable root; positive\\
quadratic factor and primal masses.\\
All assertions pass.}
\\
\addlinespace

\shortstack[l]{Raw and corrected\\continuation}
&
\shortstack[l]{All eight outcomes for 16 width pairs.\\
Closed witness and preserved\\
marginal agree.}
\\
\addlinespace

\shortstack[l]{Orientation\\certificate}
&
\shortstack[l]{All 4096 reduced eight-angle\\
assignments. Exactly 16\\
calibrated families.}
\\
\addlinespace

\shortstack[l]{General detector\\contractions}
&
\shortstack[l]{100 nonnegative measure families,\\
asymmetric outcomes, arbitrary\\
positive scales, and all feedback\\
policies. Every instance satisfies\\
the analytical lower bound.}
\\
\addlinespace

\shortstack[l]{Source\\convolution}
&
\shortstack[l]{32768-point midpoint quadrature\\
for three unequal-arm examples.\\
Agreement within $10^{-12}$.}
\\
\addlinespace

\shortstack[l]{Endpoint\\Jacobian}
&
\shortstack[l]{Central-difference diagnostic for\\
the full normalized table.\\
Rank seven.}
\\
\addlinespace

\shortstack[l]{Moment\\formulas}
&
\shortstack[l]{Direct normalized-kernel derivatives\\
at four displacements. All three\\
rational functions agree.}
\\
\addlinespace

\shortstack[l]{Positive\\five-atom family}
&
\shortstack[l]{High-precision finite sums\\
down to $t=10^{-10}$. Improvement\\
ratio approaches $\mathcal{C}_{0.4}$.}
\\
\addlinespace

\shortstack[l]{Half-circle\\response}
&
\shortstack[l]{Fourier sum with a geometric\\
tail bound. $\delta\simeq0.1743750480$\\
and $\varepsilon=0$ analytically.}
\\

\bottomrule
\end{tabular}
\end{table}

The interval certificate records the root centers, box radius, approximate inverse Jacobian, coefficient enclosure, quadratic coefficients, discriminant enclosure, and primal masses. The exact definition is the root in that box, rather than a truncated decimal coefficient. \Cref{app:certificate} explains the inclusion calculation so that the numerical certificate can be audited independently of a root finder's reported convergence.

For general numerical detectors, each setting and output is represented by a positive finite atomic measure, which is contracted directly against the continuous $S_\Gamma$. The test set includes near-calibrated and broader perturbations without imposing outcome symmetry. Every setting's positive scale is retained in the feedback weights. The source-convolution check uses the unequal widths $0.17$ and $0.23$ and does not replace their continuous convolution by a discrete sampled source.

The five-atom performance is evaluated from the actual nonlinear measure, not from its first-order expansion. At very small errors, arbitrary-precision arithmetic avoids subtractive cancellation in $[\epszero-\eps(t)]/\Delta(t)$. For the half-circle example, if the odd series is truncated before odd index $N$, a valid absolute correlation-tail bound is
\begin{equation}
\frac8{\pi^2}\frac{e^{-N\Gamma}}{N^2(1-e^{-2\Gamma})}.
\end{equation}
No experimental data or observed detector measurements enter this work. All reported probabilities are mathematical predictions of specified constructions.

\section{Physical significance and future prospects}
\label{sec:discussion}
\subsection{Adaptive composition beyond isolated no-signalling}
The central lesson is that normalized source tables omit information needed for adaptive composition. Each isolated table divides out its own partition weight. If feedback chooses different settings on different earlier-outcome branches, those formerly invisible partition weights can become relative branch weights in a single globally normalized experiment. The exact record condition identifies this mechanism without requiring a particular interpretation of the hidden state.

This matters for model development because agreement with an isolated Bell experiment is a necessary but incomplete calibration target. A detector completion must also specify its input memory, output archives, allowed controls, and boundary conditions, then preserve the accessible marginal at every relevant cut. Such requirements are experimentally motivated operations, not an optional demand that hidden variables themselves evolve forward in time.

The raw continuation witness and the terminal-detector theorem probe different failures. The former exhibits an explicit remote marginal change in a possible spacelike arrangement. The latter proves that an entire terminal response class cannot simultaneously satisfy exact source calibration and the stipulated temporal feedback consistency. The first motivates the question; the second establishes that arbitrary changes inside a detector, if they integrate to the same restricted interface, cannot solve it.

\subsection{Implications of detector rigidity and sharp asymptotic bounds}
The equal-setting source target lies on an extremal boundary of the positive response class. That fact makes a feasibility search analytically tractable: saturation fixes the detector response to antipodal atoms, leaving only orientations and scales. The minimax theorem then quantifies the unavoidable first-record bias. It is stronger than finding one malfunctioning ideal detector, because it excludes every exactly calibrated terminal detector in the stated class.

Robustness is essential to the relevance of that statement. An obstruction that disappears under arbitrarily small calibration errors could be an artifact of ideal endpoint boundaries. The concentration and harmonic-mean estimates show that this does not happen at fixed finite width. A nonzero neighborhood remains incompatible with perfect first-record preservation. The sharp-order theorem further shows that allowing a small source error only improves the optimum linearly.

At the benchmark width, the coefficient $\Csharp$ gives the best leading exchange rate between source-table error and reduction of record disturbance. It is an extremal property of the declared detector class and measures the optimal leading improvement available under its source-interface assumptions. Its significance is methodological as well as quantitative: a continuous detector-measure optimization is controlled by a uniform asymptotic reduction and certified by matching positive and polynomial constructions. The root system provides a reproducible definition, and the polynomial certificate controls the full continuous response domain.

\subsection{Memory, hidden dependence, and the statistical arrow of records}
The protected-tag construction shows a genuine compatibility: accessible records can be preserved while hidden interior bridges remain dependent on future settings. Thus an operational arrow for records need not be imposed on every hidden variable. Equally, the passive-memory obstruction shows why simply copying a result does not generate that compatibility. The weight of the recorded branch must be controlled by the apparatus rule, and that control can depend on preparation information.

The reset examples distinguish a visible display from the information used by subsequent detector interactions. Erasing a display while retaining preparation data in an environment is not the same operation as removing the preparation data while retaining an independent archive. A proposed physical theory must assign both operations explicit dynamics and account for all their outputs. It cannot infer statistical preservation solely from the fact that an old bit still exists somewhere.

The shared-hidden-factor example adds a second caution: a cancellation valid under one hidden ensemble can fail after a previous device changes that ensemble. The hidden-state update is therefore part of the composition interface. One-device normalization tests are insufficient for repeated use unless their normalization is valid for every reachable interface state.

\subsection{Structural scope of the detector obstruction}
\Cref{tab:scope} collects the assumptions attached to each result.

\begin{table}[tbp]
\centering
\footnotesize
\setlength{\tabcolsep}{3pt}
\renewcommand{\arraystretch}{1.08}

\caption{Results and the structural assumptions defining
their domains.}
\label{tab:scope}

\begin{tabular}{@{}ll@{}}
\toprule
Problem & Assumptions and result \\
\midrule

\shortstack[l]{Raw adaptive\\continuation}
&
\shortstack[l]{Finite positive widths, literal\\
segment weights, and global\\
normalization. Exact\\
remote-marginal witness.}
\\
\addlinespace

\shortstack[l]{Conditional\\continuation}
&
\shortstack[l]{Finite local measurement trees\\
with normalized transitions.\\
Prefix and transcript preservation.}
\\
\addlinespace

\shortstack[l]{Common terminal\\detector}
&
\shortstack[l]{Arrival-angle source interface\\
and (D1)--(D4). Rigidity and\\
sharp exact record bias.}
\\
\addlinespace

\shortstack[l]{Imperfect\\calibration}
&
\shortstack[l]{Nonnegative response measures\\
at fixed finite source width.\\
Finite-error bounds and linear\\
small-error improvement.}
\\
\addlinespace

\shortstack[l]{Sharp\\coefficient}
&
\shortstack[l]{Four calibrated settings at\\
$\Gamma=0.4$. Full measure-class\\
optimum to first order, attained\\
by a five-atom family.}
\\
\addlinespace

\shortstack[l]{Memory\\mechanism}
&
\shortstack[l]{Explicit preparation tag and\\
normalized continuation sector.\\
Record preservation with hidden\\
future dependence; reset\\
and reuse criteria.}
\\
\addlinespace

Reciprocity
&
\shortstack[l]{Fixed normalized transition kernel.\\
Endpoint reciprocity.}
\\

\bottomrule
\end{tabular}
\end{table}

A fixed source reference measure does not assume away hidden retrocausality; the normalized source posterior can remain setting dependent. Bell locality appears only in particular conclusions or alternative constructions, such as uniform hidden-leg normalization and constant-total response examples. Likewise, (D1)--(D4) should not be replaced by the vague phrase ``one detector law.'' One microscopic interaction can receive different physical input states, and those inputs can invalidate the terminal reduction used in the theorem.

The theorem therefore constrains a precise source interface and composition law. The positive continuation constructions identify how preparation information and branch normalization preserve observable records while retaining hidden future dependence. The adaptive conditional reconstruction of Ref.~\cite{wharton2020} implements the corresponding alternative branch prescription, providing a direct point of comparison for physical apparatus extensions.

\subsection{Prospects for a physical apparatus theory}
A natural next investigation is a source--controller--detector model with an explicit preparation input, pointer, and environment. The new input must have a defined state space, transport law, interaction, and boundary condition. Merely adding an independent random register that integrates to another allowed terminal measure does not change the theorem class. To escape the obstruction, the extension must alter an actual interface or composition assumption.

The causal partner experiment provides a useful target. A controller can in principle carry the earlier setting and outcome as ordinary forward information; the effective repair shows that the relevant source-path width also matters. The open task is to derive the required integrated branch masses from one physical interaction, without inserting a normalization factor only after recognizing the overall protocol. Every pointer outcome, time record, failure outcome, and reset output must be retained. An implementation that conditions on a successful detector event would require a separate analysis of the rejected events.

A successful apparatus model would be important beyond repairing this particular example. It would demonstrate how future-dependent hidden variables can coexist with a stable statistical history of macroscopic records under reuse and feedback. A further impossibility theorem for a larger, physically specified interface would also be informative: it would identify which preparation, environmental, or source structures are genuinely necessary.

\subsection{Prospects for broader composition and quantitative tests}
Several mathematical questions are now well posed. The exact finite-error curve could be approached by bounds that retain correlations between diagonal defects, oblique errors, and marginal imbalances, rather than bounding them separately. The continuous moment functions \cref{eq:a0,eq:a1,eq:g} suggest studying how the number and positions of optimal contacts change with $\Gamma$. A coefficient for a continuum of calibrated settings would require all angular error constraints and cannot be inferred from the four-setting certificate.

A network extension should retain the complete hidden interface, including environmental correlations and every controller memory used later. Two-source joint measurements and entanglement swapping would then test whether the update law remains closed under a new class of devices. The restricted Gram result in \cref{app:gram} illustrates how adaptive preservation can force locality under additional symmetry assumptions; its assumption-boundary example shows why that conclusion does not extend automatically.

The bounds can also guide eventual empirical model comparison. At the level of exact probabilities, any realization satisfying the declared structural assumptions and calibrated within $\delta$ must have first-record disturbance at least $F_{\rm new}(\delta)$. An experiment finding both better calibration and a smaller disturbance would reject that conjunction of assumptions. Turning this observation into a statistical protocol requires finite-sample confidence regions, an operational model of setting implementation, timing and no-click records, and justification of the detector interface. These steps define a concrete route from the present probability bounds to experimental tests of the specified source--detector architecture.

\section{Conclusion}
A finite-width Schulman-type source can possess uniform isolated marginals while a literal raw-weight extension fails under classical feedback. The integrated future row masses determine whether earlier accessible records remain statistically unchanged. Normalized conditional continuations and protected preparation memory provide restricted positive constructions, but passive copying, tag-only reset, and adaptive first detection reveal their limits.

Within an explicit terminal-detector class, exact source calibration forces equal-mass antipodal responses and entails the sharp record disturbance $[2\cosh(2\Gamma)]^{-1}$. The obstruction persists under imperfect calibration, the optimal near-zero improvement is linear at every fixed positive width, and a continuous primal--dual certificate determines its coefficient at $\Gamma=0.4$ for four settings. The resulting expansion is attained by a positive five-atom detector.

These findings establish adaptive record preservation as a quantitative design constraint for retrocausal models. The exact obstruction, stable finite-error bounds, and attained sharp coefficient provide a benchmark for apparatus theories that include preparation information, hidden-state updates, and environmental records. This framework supports a systematic programme of physically specified detector extensions, repeated-use protocols, and network composition tests.

\paragraph{Data and code availability.}
The accompanying source package contains the editable manuscript, bibliography, vector figures, verification programs, and their numerical and interval-certificate outputs. The reported results are theoretical and computational. All figures and numerical tables are reproducible from the supplied programs.

\paragraph{AI-assisted research and manuscript preparation.}
The author used OpenAI's ChatGPT to assist with implementation of numerical verification code, and manuscript preparation. The author takes responsibility for the final arguments, computations, interpretations, and conclusions.

\FloatBarrier
\bibliographystyle{quantum}
\bibliography{references}

@article{almada2016,
 author={Almada, D. and Ch'ng, K. and Kintner, S. and Morrison, B. and Wharton, K. B.},
 title={Are Retrocausal Accounts of Entanglement Unnaturally Fine-Tuned?},
 journal={International Journal of Quantum Foundations}, volume={2}, pages={1--16}, year={2016},
 url={https://arxiv.org/abs/1510.03706}, note={arXiv:1510.03706v1}}

@article{wharton2020,
 author={Wharton, K. B. and Argaman, N.},
 title={Colloquium: {Bell}'s theorem and locally mediated reformulations of quantum mechanics},
 journal={Reviews of Modern Physics}, volume={92}, pages={021002}, year={2020},
 doi={10.1103/RevModPhys.92.021002}, url={https://arxiv.org/abs/1906.04313}, note={arXiv:1906.04313v3}}

@article{argaman2026,
 author={Argaman, Nathan}, title={Ontic and epistemic states in the theory of spacetime-local beables},
 journal={Entropy}, volume={28}, number={6}, pages={584}, year={2026},
 doi={10.3390/e28060584}, url={https://arxiv.org/abs/2609.00848}}

@article{allcock2009,
 author={Allcock, Jonathan and Brunner, Nicolas and Linden, Noah and Popescu, Sandu and Skrzypczyk, Paul and V{\'e}rtesi, Tam{\'a}s},
 title={Closed sets of nonlocal correlations}, journal={Physical Review A}, volume={80}, pages={062107}, year={2009},
 doi={10.1103/PhysRevA.80.062107}, url={https://arxiv.org/abs/0908.1496}}

@misc{selby2022,
 author={Selby, John H. and Stasinou, Maria E. and Gogioso, Stefano and Coecke, Bob},
 title={Time symmetry in quantum theories and beyond}, year={2022}, note={\href{https://doi.org/10.48550/arXiv.2209.07867}{arXiv:2209.07867v2}, revised 2024},
 doi={10.48550/arXiv.2209.07867}, url={https://arxiv.org/abs/2209.07867}}

@article{wood2015,
 author={Wood, Christopher J. and Spekkens, Robert W.},
 title={The lesson of causal discovery algorithms for quantum correlations: Causal explanations of {Bell}-inequality violations require fine-tuning},
 journal={New Journal of Physics}, volume={17}, pages={033002}, year={2015},
 doi={10.1088/1367-2630/17/3/033002}, url={https://arxiv.org/abs/1208.4119}}

@misc{brogioli2024,
 author={Brogioli, Doriano}, title={A no-go theorem for sequential and retro-causal hidden-variable theories based on computational complexity},
 year={2024}, note={\href{https://doi.org/10.48550/arXiv.2409.11792}{arXiv:2409.11792}}, doi={10.48550/arXiv.2409.11792}, url={https://arxiv.org/abs/2409.11792}}

@article{price2021,
 author={Price, Huw and Wharton, Ken}, title={Entanglement Swapping and Action at a Distance},
 journal={Foundations of Physics}, volume={51}, pages={105}, year={2021},
 doi={10.1007/s10701-021-00511-3}, url={https://arxiv.org/abs/2101.05370},
 note={The extended arXiv:2101.05370v4 includes the appendix discussed here}}
\onecolumn
\newpage
\appendix
\section{Kernel identities and normalization details}
\label{app:identities}
For $\nu>0$, the Poisson-kernel expansion of \cref{eq:kernel} is
\begin{equation}
k_\nu(d)=\frac1{2\pi}\left(1+2\sum_{n\ge1}e^{-n\nu}\cos nd\right).
\end{equation}
Absolute convergence permits termwise convolution, which multiplies Fourier coefficients and proves \cref{eq:convolution}. Antipodal summation removes the odd harmonics and yields \cref{eq:R}. The extrema are $k_\nu(0)=\coth(\nu/2)/(2\pi)$ and $k_\nu(\pi)=\tanh(\nu/2)/(2\pi)$.

For the source, the endpoint difference under the antipodal initial constraint is congruent to $x-y$ when $ab=-1$ and to $x-y+\pi$ when $ab=+1$. Hence the four-outcome normalizer is $2[k_\Gamma(x-y)+k_\Gamma(x-y+\pi)]$. The identity
\begin{equation}
\frac{k_\Gamma(d)-k_\Gamma(d+\pi)}{k_\Gamma(d)+k_\Gamma(d+\pi)}
=\sech\Gamma\cos d
\end{equation}
proves the target law. Applied to a same-particle segment, it gives the positive-sign transition correlation in \cref{eq:transition}.

The source posterior \cref{eq:hidden-source} is obtained by summing each leg's two endpoint weights before normalizing the $\lambda$ integral. Its setting dependence is therefore present despite a uniform reference measure. At an interior continuation cut, convolution with the endpoint sum gives
\begin{equation}
\int k_{\nu_1}(\phi-\theta(q,b))R_{\nu_2}(z-\phi)\dd\phi=R_{\nu_1+\nu_2}(z-q),
\end{equation}
where antipodal periodicity removes dependence on the starting sign in the denominator. This explicitly verifies the averaged record condition for \cref{eq:bridge}.

\section{Derivative formulas for the moment certificate}
\label{app:derivatives}
A direct route to \cref{eq:a0,eq:a1,eq:g} starts with a symmetric tail of unit infinitesimal mass at displacement $\alpha$:
\begin{equation}
\nu_t=(1-t)\delta_0+\tfrac t2(\delta_\alpha+\delta_{-\alpha}).
\end{equation}
Let $s(d)=C_\Gamma/(c+\cos d)$ and
\begin{equation}
H_\alpha(d)=s(d+\alpha)+s(d-\alpha)-2s(d).
\end{equation}
Then $K_t(d)=s(d)+tH_\alpha(d)+O(t^2)$. Put
\begin{equation}
J_t(d)=\tfrac12[K_t(d)+K_t(d+\pi)],\qquad
E_t(d)=\frac{K_t(d)-K_t(d+\pi)}{K_t(d)+K_t(d+\pi)}.
\end{equation}
Differentiation at $t=0$ gives
\begin{equation}
\dot E(d)=\frac{2[s(d+\pi)H_\alpha(d)-s(d)H_\alpha(d+\pi)]}
{[s(d)+s(d+\pi)]^2}.
\end{equation}
Thus $a_0(\cos\alpha)=\dot E(0)/2$ and $a_1(\cos\alpha)=\dot E(\pi/4)/2$. For the disturbance
\begin{equation}
\eps(t)=\frac{J_t(0)-J_t(\pi/2)}{2[J_t(0)+J_t(\pi/2)]}
\end{equation}
near the ideal detector, the positive improvement derivative is
\begin{equation}
g(\cos\alpha)=-\dot\eps(0)
=\frac{J_0(0)\dot J(\pi/2)-J_0(\pi/2)\dot J(0)}{[J_0(0)+J_0(\pi/2)]^2}.
\end{equation}
Substituting the rational kernel and $\cos\alpha=z$ yields the displayed moment functions. Linearity gives the first variation for a general even measure; the uniform weak-norm remainder is supplied by \cref{eq:bilinear-remainder}. This last step is what permits small diffuse responses as well as finite tails.

For the full finite five-atom response with locations $\xi_i$ and weights $w_i(t)$, use the exact finite sums
\begin{equation}
K_t(d)=\sum_{i,j}w_i(t)w_j(t)\frac{C_\Gamma}{c+\cos(d+\xi_i-\xi_j)}.
\end{equation}
The calibration and disturbance on the stated alphabet are
\begin{equation}
\Delta(t)=\frac12\max_{d\in\{0,\pi/4,\pi/2\}}|E_t(d)+\sech\Gamma\cos d|,
\qquad
\eps(t)=\frac{\max_dJ_t(d)-\min_dJ_t(d)}{2[\max_dJ_t(d)+\min_dJ_t(d)]}.
\end{equation}
For equal scales and rotation-covariant responses, every first setting has access to the same set of separations, which justifies the final extremal-ratio expression.

\section{Explicit root box and interval proof procedure}
\label{app:certificate}
The coefficient is defined without relying on rounded parameter values. Let $F$ be the four-component function in \cref{eq:root-system}, in the variable order $(\lambda_0,\lambda_1,r_1,r_2)$. Each coordinate interval has radius $10^{-45}$ about the following exact terminating decimal center:
\begin{center}\small
\begin{tabular}{@{}cl@{}}\toprule
Coordinate & Decimal center\\\midrule
$\lambda_0$ & \texttt{8.61244939901755559576260020598322651327701659030544}\\
$\lambda_1$ & \texttt{0.73864237203945398769150565486663373104609701033438}\\
$r_1$ & \texttt{0.71879430577723234164553401714053171047810388392274}\\
$r_2$ & \texttt{0.13011536807962724246635804934153974195400657439379}\\\bottomrule
\end{tabular}
\end{center}
The machine-readable certificate retains longer centers and tighter evaluation detail. The shortened centers above change the box center by less than $10^{-50}$; the same interval-inclusion test also verifies this displayed box.

Here is the validation procedure. Let $X$ denote the box, $x_0$ its center, and $A$ a high-precision approximate inverse of $F'(x_0)$. Treat every stored entry of $A$ as an exact terminating decimal when evaluating intervals. Bound the derivative $F'(X)$ by interval polynomial evaluation and form
\begin{equation}
\mathcal K(X)=x_0-AF(x_0)+[I-AF'(X)](X-x_0).
\label{eq:interval-map}
\end{equation}
The script verifies strict inclusion $\mathcal K(X)\subset\operatorname{int}X$ and an infinity-norm bound below one for $I-AF'(X)$. The map $x\mapsto x-AF(x)$ maps the closed box into itself and is a contraction there. Its unique fixed point is a root of $F$, since $A$ is nonsingular. The exact coefficient is the sum of its first two coordinates.

All entries of the derivative can be computed polynomially. If $\mathcal P=\lambda_0D+\lambda_1E-N$, the two rows at contact $r_1$ are
\begin{equation}
\bigl(D(r_1),E(r_1),\mathcal P'(r_1),0\bigr),\qquad
\bigl(D'(r_1),E'(r_1),\mathcal P''(r_1),0\bigr),
\end{equation}
and the rows at $r_2$ have the corresponding derivatives in the fourth column. The constants $c=\cosh(0.4)$ and $\sqrt2$ are themselves enclosed by interval arithmetic; they are not replaced by machine-precision constants.

Divide the degree-six polynomial by $(z-r_1)^2(z-r_2)^2$ using interval polynomial arithmetic to enclose the quotient coefficients. Exact divisibility follows from the root equations and distinct contacts. The computed quotient bounds $q_2>0$ and $q_1^2-4q_0q_2<-44.7812$ prove strict positivity of the quadratic for every real $z$. Finally, interval evaluation checks $\lambda_0,\lambda_1,\eta_1,\eta_2>0$ and encloses their coefficient sum. These steps certify the continuous dual inequality and a matching positive primal construction.

\section{Additional constructive calibration families}
\label{app:families}
The earlier robustness comparison can be reproduced by mixing an ideal endpoint response with the cosine response. For $0\le r\le1$, set
\begin{equation}
\mu^{(r)}_{x,a}=\frac{1-r}{2}\delta_{\theta(x,a)}+r\mu^{\cos}_{x,a}.
\end{equation}
With $d=x-y$, $v=\sech\Gamma$, $q=e^{-\Gamma}$, and
\begin{equation}
J(d)=\frac{C_\Gamma c}{c^2-\cos^2d},\qquad z_c=\frac1{2\pi},
\end{equation}
direct contraction gives
\begin{align}
z_r(d)&=(1-r)^2J(d)+(2r-r^2)z_c,\\
E_r(d)&=-\cos d\,
\frac{(1-r)^2J(d)v+z_cq(2r-\tfrac32r^2)}{z_r(d)}.
\end{align}
The marginals are uniform. Its calibration error is the largest half-correlation error over the alphabet, and its disturbance is
\begin{equation}
\eps_r=\frac{z_r(0)-z_r(\pi/2)}{2[z_r(0)+z_r(\pi/2)]}.
\end{equation}
An oblique pair can determine the calibration error, so checking only the diagonal would be insufficient. This family is an admissible interpolation, not a full-class optimization.

An earlier near-optimal finite-tail family uses
\begin{equation}
\nu_t=[1-(b_1+b_2)t]\delta_0+
\frac{b_1t}{2}(\delta_{\alpha_1}+\delta_{-\alpha_1})+
\frac{b_2t}{2}(\delta_{\alpha_2}+\delta_{-\alpha_2}),
\end{equation}
with exact decimal parameters
$\alpha_1=0.7686$, $\alpha_2=1.4396$, $b_1=4.7181$, $b_2=5.7940$.
Its derivative formulas yield $\Delta'(0+)\simeq0.9999934792$ and $-\eps'(0+)\simeq9.3509191479$, so the improvement-to-error limit is approximately $9.3509801237$. This supplies a concrete improvement over smoothing independently of any optimality claim. The root-defined family in \cref{eq:fiveatom} replaces it for the sharp asymptotic theorem. Asymptotic optimality does not require it to outperform every other family at every nonzero error.

\section{Gram-matrix constraints on adaptive composition}
\label{app:gram}
This result concerns a broader response representation, but uses additional symmetry. Let equal effective arms have nonnegative functions $f_x(a,\lambda)$ and total envelopes $F_x(\lambda)=\sum_af_x(a,\lambda)$ with finite square integrals under a fixed measure $\mu$. Assume $F_x(\lambda+\pi)=F_x(\lambda)$ and exact uniform isolated marginals. Then the source partition matrix is the positive Gram matrix
\begin{equation}
Z_{xy}=\int F_x(\lambda)F_y(\lambda)\mu(\dd\lambda).
\end{equation}
\begin{proposition}[Gram-envelope collapse]
If setting-only raw feedback preserves every first marginal $1/2$ on a common setting alphabet including diagonal experiments, all $F_x$ agree almost everywhere, and the isolated model is Bell-local.
\end{proposition}
\begin{proof}
Feedback branch probabilities are $Z_{x,y_+}/(Z_{x,y_+}+Z_{x,y_-})$. Preservation for independently selected branch settings makes each row of $Z$ constant. Symmetry makes that constant common to all rows. Thus
\begin{equation}
\norm{F_x-F_y}_{L^2(\mu)}^2=Z_{xx}+Z_{yy}-2Z_{xy}=0.
\end{equation}
Write the common function as $F$ and $Z=\int F^2\dd\mu$. On $F>0$,
\begin{equation}
P(a,b\mid x,y)=\int\frac{F(\lambda)^2\mu(\dd\lambda)}{Z}
\frac{f_x(a,\lambda)}{F(\lambda)}
\frac{f_y(b,\lambda+\pi)}{F(\lambda)}.
\end{equation}
The latter factors are normalized local responses under a setting-independent hidden distribution. The zero-$F$ set has no weight.
\end{proof}

Antipodal evenness cannot simply be removed. Under uniform circle measure, take
\begin{equation}
F_x(\lambda)=1+\zeta t_x(\cos\lambda+\cos2\lambda),\qquad |\zeta t_x|<1/2.
\end{equation}
These strictly positive envelopes need not agree, but
\begin{equation}
\int_\T F_x(\lambda)F_y(\lambda+\pi)\frac{\dd\lambda}{2\pi}=1
\end{equation}
because odd and even harmonic contributions cancel. This is an envelope counterexample to the unrestricted Gram inference, not a full detector outcome model. In the exactly calibrated fixed-kernel class, \cref{thm:rigidity} already forces the relevant evenness.

\section{Computational verification of the uniform reduction}
\label{app:audit}
The uniform reduction is supported by the explicit weak-norm estimate \cref{eq:gauge-explicit-bound} and the direct witness expansion \cref{eq:w-explicit-expansion}. This appendix gives computational checks of the gauge structure, asymmetric detector perturbations, and the resulting first-order optimum.

The gauge argument uses the full four-setting source tables. Their diagonal entries improve the mass imbalance from an initial $O(\sqrt\delta)$ estimate to $O(\delta)$. Orthogonal pairs remove relative shifts between the two outputs and between orthogonal settings; an oblique pair removes the remaining relative shift between the orthogonal pairs. The only angular null direction is a common physical rotation. The explicit analytic Jacobian has rank seven; a numerical evaluation of that analytic matrix gives seven nonzero singular values, the smallest approximately $0.654079$, and a rotation-null residual below $10^{-14}$. This numerical rank check supplements the elementary kernel proof.

A separately implemented finite-history calculation tests two asymmetric families: a concentrated response with weak displaced tails, and a response whose split width is proportional to $\sqrt t$. Both include independent order-$t$ drifts, outcome-mass imbalances, a removable common rotation, unequal baseline scales between orthogonal pairs, and perturbed within-pair scales. Twelve instances use $t$ from $10^{-3}$ to $3\times10^{-6}$. All sixteen dihedral transforms are evaluated. The weak-norm bound divided by actual calibration error remains bounded, and the table-averaging and witness-averaging remainders divided by $t^2$ approach finite values. The nonzero second-order remainders agree with the uniform first-order averaging statement used in the theorem.

A further sampled tangent program allows independent jump directions, drifts, positive diffusion directions, outcome-mass imbalances, and setting scales at every setting and outcome. Its 561 variables and 256 inequalities include every first-order source total-variation constraint and every active first-record comparison. The displacement set contains a uniform grid and the certified contact displacements. It returns an improvement coefficient $9.351091771057005$, with maximum constraint violation below $7.7\times10^{-14}$. This independently implemented finite program agrees with the coefficient established over the continuous measure class by the analytical reduction and polynomial certificate.

The programs \texttt{audit\_reduction.py} and \texttt{audit\_tangent.py} and their result files accompany the manuscript. The separate implementations provide reproducible checks of the analytical reduction and its sharp coefficient; the uniform estimates and validated continuous polynomial certificate establish the full measure-class result.

\end{document}